\documentclass[runningheads]{llncs}
\RequirePackage[margin=1.0in]{geometry} 

\RequirePackage[american]{babel}
\RequirePackage{csquotes}
\RequirePackage[T1]{fontenc} 

\RequirePackage{quoting}
\newenvironment{emphasized}{%
    \begin{quoting}[font={itshape}]
}{%
    \end{quoting}
} 

\RequirePackage{amsmath, mathtools, mathrsfs}
\RequirePackage{amssymb}

\RequirePackage{thm-restate} 

\newcommand{\eqdef}{\stackrel{\text{\tiny def}}{=}} 
\newcommand{\from}{\colon} 

\DeclarePairedDelimiter{\xpar}{\lparen}{\rparen} 

\providecommand{\given}{}
\DeclarePairedDelimiterX{\set}[1]{\{}{\}}{%
  \renewcommand{\given}{\nonscript\;\delimsize\vert\nonscript\;\mathopen{}} #1%
} 

\DeclareMathOperator{\interior}{int} 

\newcommand{\one}{\mathbf{1}} 

\DeclareMathOperator{\supp}{supp} 

\providecommand{\given}{}
\DeclarePairedDelimiterX{\meanbrackets}[1]{\lbrack}{\rbrack}{%
  \renewcommand{\given}{\nonscript\,\delimsize\vert\nonscript\,\mathopen{}} #1%
}
\DeclareMathOperator*{\meansymbol}{\mathbb{E}}
\NewDocumentCommand{\mean}{e{_}}{%
  \meansymbol%
  \IfNoValueF{#1}{\sb{#1}}%
  \meanbrackets%
} 

\newcommand{\param}[1]{%
  \texorpdfstring{\ensuremath{#1}\protect\nobreakdash}{#1}%
} 

\newcommand{\gradient}{\nabla} 

\DeclarePairedDelimiter{\abs}{\lvert}{\rvert} 

\DeclareMathOperator{\tangent}{\mathcal{T}} 

\DeclarePairedDelimiterX{\inner}[2]{\langle}{\rangle}{#1,\,#2} 
\DeclarePairedDelimiter{\norm}{\lVert}{\rVert} 

\DeclareMathOperator{\nash}{NE} 

\newcommand{\naturals}{\mathbb{N}} 

\newcommand{\reals}{\mathbb{R}}    

\let\nnreals\nonnegativereals      

\newcommand{\condiffs}{%
  \mathcal{C}_{1}}                 

\let\polys\polynomials             

\let\sospolys\sumofsquares         
\newcommand{\sos}{\sigma}          

\newcommand{\normal}{\mathcal{N}}  

\newcommand{\game}{\mathscr{G}}    
\newcommand{\noplayers}{n}         
\newcommand{\pures}{\mathcal{A}}   
\newcommand{\nopures}{m}           
\newcommand{\pure}{s}              
\newcommand{\payoff}{u}            
\newcommand{\instantpayoff}{v}     
\newcommand{\payoffmat}{%
  \mathbf{A}}                      
\newcommand{\actions}{\mathcal{X}} 
\newcommand{\action}{x}            
\newcommand{\policy}{f}            
\newcommand{\uncoupleds}{%
  \mathcal{F}}                     
\newcommand{\uncoupled}{%
  \policy^{\mathrm{un}}}           
\newcommand{\horizon}{\reals}      
\newcommand{\initialaction}{%
  \action_{0}}                     
\newcommand{\flow}{\phi}           
\newcommand{\explorerate}{T}       

\newcommand{\nodata}{K}            
\newcommand{\semiset}{\mathcal{S}} 
\newcommand{\ngcon}{g}             
\newcommand{\ngcons}{\mathcal{G}}  
\newcommand{\eqcon}{h}             
\newcommand{\eqcons}{\mathcal{H}}  
\newcommand{\perms}{\mathcal{P}}   

\newcommand{\siarsol}{p_{%
  \mathrm{\glsentryshort{SIAR}}}}  
\newcommand{\sindysol}{p_{%
  \mathrm{\glsentryshort{SINDy}}}} 

\newcommand{\mle}{\lambda_{\max}}  

\RequirePackage{graphicx}
\graphicspath{{figures/}}
\DeclareGraphicsExtensions{.pdf, .png}

\RequirePackage{tikz}
\usetikzlibrary{calc, positioning, fit}

\RequirePackage[
  abbreviations,
  automake,
  nogroupskip, 
  style=long
]{glossaries-extra}
\makeglossaries

\GlsXtrEnableEntryCounting{abbreviation}{0} 

\GlsXtrEnableInitialTagging{abbreviation}{\initial}

\RequirePackage[
  style = numeric-comp,
  sorting = none,
  maxcitenames = 1,
  arxiv = abs
]{biblatex}
\RequirePackage{hyperref, xcolor}
\definecolor{MyLinkColor}{HTML}{800006}
\definecolor{MyCiteColor}{HTML}{2E7E2A}
\definecolor{MyFileColor}{HTML}{131877}
\definecolor{MyURLColor}{HTML}{8A0087}
\definecolor{MyMenuColor}{HTML}{727500}
\definecolor{MyRunColor}{HTML}{137776}
\colorlet{MyLinkBorderColor}{MyLinkColor!60!white}
\colorlet{MyCiteBorderColor}{MyCiteColor!60!white}
\colorlet{MyFileBorderColor}{MyFileColor!60!white}
\colorlet{MyURLBorderColor}{MyURLColor!60!white}
\colorlet{MyMenuBorderColor}{MyMenuColor!60!white}
\colorlet{MyRunBorderColor}{MyRunColor!60!white}
\hypersetup{
  colorlinks,
  linkcolor=MyLinkColor,
  citecolor=MyCiteColor,
  filecolor=MyFileColor,
  urlcolor=MyURLColor,
  menucolor=MyMenuColor,
  runcolor=MyRunColor,
  linkbordercolor=MyLinkBorderColor,
  citebordercolor=MyCiteBorderColor,
  filebordercolor=MyFileBorderColor,
  urlbordercolor=MyURLBorderColor,
  menubordercolor=MyMenuBorderColor,
  runbordercolor=MyRunBorderColor
} 

\RequirePackage[
  noabbrev,
  capitalise
]{cleveref}
\RequirePackage{crossreftools}

\RequirePackage[inline]{enumitem} 
\setlist[enumerate,1]{label={\arabic*)}}

\RequirePackage{booktabs}

\RequirePackage{siunitx}
\newabbreviation{CS}{CS}{%
  \initial{C}auchy–\initial{S}chwarz inequality
}

\newabbreviation[
  description={\initial{M}aximal-\initial{L}yapunov-\initial{E}xponent}
]{MLE}{MLE}{%
  maximal Lyapunov exponent
}

\newabbreviation[
  plural={Quantal Response Equilibria}
]{QRE}{QRE}{%
  \initial{Q}uantal \initial{R}esponse \initial{E}quilibrium
}

\newabbreviation[
  description={\initial{R}ock-\initial{P}aper-\initial{S}cissors}
]{RPS}{RPS}{%
  rock-paper-scissors
}

\newabbreviation{SIAR}{SIAR}{%
  \initial{S}ide-\initial{I}nformation \initial{A}ssisted \initial{R}egression
}

\newabbreviation{SINDy}{SINDy}{%
  \initial{S}parse \initial{I}dentification of \initial{N}onlinear \initial{Dy}namics
}

\newabbreviation[
  description={\initial{S}um \initial{o}f \initial{S}quares}
]{SOS}{SOS}{%
  sum-of-squares
}

\RequirePackage{lineno}  

\title{Data-Scarce Identification of Game Dynamics \\ via Sum-of-Squares Optimization}

\author{%
  Iosif Sakos\inst{1} \and %
  Antonios Varvitsiotis\inst{1} \and %
  Georgios Piliouras\inst{1}%
}
\authorrunning{I. Sakos et al.} 
\institute{%
  Singapore University of Technology and Design, Singapore 487372, Singapore\\%
  \email{\{iosif\_sakos, antonios, georgios\}@sutd.edu.sg}%
}

\begin{document}

\maketitle

\begin{abstract}


Understanding how players adjust their strategies in games, based on their experience, is a crucial tool for policymakers.
It enables them to forecast the system's eventual behavior, exert control over the system, and evaluate counterfactual scenarios. 
The task becomes increasingly difficult when only a limited number of observations are available or difficult to acquire. 
In this work, we introduce the \gls{SIAR} framework, designed to identify game dynamics in multiplayer normal-form games only using data from a short run of a single system trajectory. 
To enhance system recovery in the face of scarce data, we integrate side-information constraints into \gls{SIAR}, which restrict the set of feasible solutions to those satisfying game-theoretic properties and common assumptions about strategic interactions. 
\Gls{SIAR} is solved using \gls{SOS} optimization, resulting in a hierarchy of approximations that provably converge to the true dynamics of the system. 
We showcase that the \gls{SIAR} framework accurately predicts player behavior across a spectrum of normal-form games, widely-known families of game dynamics, and strong benchmarks, even if the unknown system is chaotic. 
\end{abstract}
\glsresetall 

\section{Introduction} 
\label{sec:Introduction}


Game theory is the study of strategic interaction between multiple self-interested decision makers. 
Whilst the main solution concept is that of a Nash equilibrium, this equilibrium point of view is nowadays well understood to be limited in terms of its real world applicability due to at least three different type of problems. 
The first problem is that due to the multiplicity of the Nash equilibria the equilibrium selection problem is still a largely unsolved problem at the intersection of behavioral game theory and game dynamics \parencite{harsanyi_new_1995, bolle_behavioral_2017, cronert_equilibrium_2022}.
The second problem is the computational intractability of Nash equilibria \parencite{daskalakis_complexity_2009}. 
This of course raises an important criticism for the Nash equilibrium as a solution concept captured by Kamal Jain’s quote: \enquote{If your laptop can’t find it then neither can the market}. 
Lastly, even if we allowed for multiagent learning dynamics to run forever strong impossibility results preclude convergence to Nash equilibria even asymptotically \parencite{hart_uncoupled_2003, milionis_impossibility_2023}. 
In fact, not only do standard learning dynamics such as replicator dynamics, multiplicative weights update, gradient descent, a.o., do not converge to Nash equilibria  even in standard classes of games such as zero-sum games, but they can even be formally chaotic \parencite{sato_chaos_2002, bailey_multiplicative_2018, cheung_vortices_2019, cheung_evolution_2022}. 
The last type of obstacle, i.e., the possible disagreement of game dynamics and Nash equilibria also suggests a different possible way forward that has been most clearly advocated by Fields Medal winner Stephen Smale, recognized for his seminal contributions in the field of dynamical systems. 
\Textcite{smale_dynamics_1976} argues that \enquote{equilibrium theory is far from satisfactory} since it ignores how equilibria are reached and assumes that agents can solve complex optimization problems. 
Instead, he proposes to study the same multiagent problems from the lens of differential equations.

In this paper, we are taking exactly this point of view and study strategic games from the lens of the dynamical behavior of the participating agents. 
All strategic multiplayer systems that we will be studying possess the following two key characteristics.
Firstly, players have their own preferences for each state of the system, which can be influenced by the strategies of other players. 
This is the static part of the system encoded by payoff matrices/functions. 
This part we will consider fixed and known. 
Intuitively, since this part is static we presume that we can always recover it to arbitrary accuracy by a one-off offline questioning of the participating agents.
This assumption is also in great agreement with the standard game theoretic assumption that considers the payoff structure of the game to be common knowledge amongst all agents.  
Secondly, players have the autonomy to take strategies based on their own preferences and to do so they learn from their experience and adjust their strategies based on the history of play.
This is the dynamic part of our setting and our goal would be to accurately identify these behavioral idiosyncrasies of the learning agents, e.g., how myopically do the agents learn? 
As policymakers are often able to observe the evolution of such systems for only a short period of time before it is required to act upon their knowledge we will be particularly driven by the following data scarcity consideration:  
\begin{emphasized}
  Can we discover the dynamics of a multiagent system by observing it for a short period of time?
\end{emphasized}
As a concrete example let's consider a traffic network where drivers choose their daily route based on their preference for less congestion. 
Moreover, in reaction to congestion on a particular route, a driver may decide to switch to a different route. 
The static part of the network congestion games is actually the latency function on different roads based on their current load, typically captured within the game theoretic literature by (network) congestion games \parencite{rosenthal_class_1973} and  methodological tools already exist for extracting these cost functions from actual traffic data \parencite{zhang_price_2016, monnot_routing_2022}.
Taking this game description as given one could ask if we take a few snapshots of current traffic estimates can we predict the future evolution of the system?

The modern approach of dynamical system discovery relies on vast amounts of available measurements, coupled with powerful machine learning algorithms and inexpensive parallel processing power \parencite{brunton_data-driven_2019}.
Important approaches that have received significant attention include symbolic regression and techniques relying on sparsity-promoting optimization \parencite{icke_modeling_2013, brunton_discovering_2016, rudy_data-driven_2017, kaiser_sparse_2018}.
Nevertheless, techniques that require access to massive data sets have limited applicability in settings where data are scarce, and expensive to acquire, or in time-critical applications where there is not enough time to collect data.  
To enable effective policymaking and decision-making, it is crucial to develop tools capable of identifying the game dynamics that govern the evolution of players' behavior, even with limited real-time data.
We focus on continuous-time time-invariant game dynamics.
In contrast to these prior works, we focus on the problem of identifying the ground truth dynamics using only a very small number of samples, typically, merely five samples, of a single system trajectory.


\subsubsection*{Our Contributions}
\label{sec:Contributions}
In this work, we introduce the \gls{SIAR} method, a system identification framework that is tailored for game dynamics.
\Gls{SIAR} is designed to discover the dynamical behavior of a multiagent system in real-time using limited observational data. 
Our focus is on strategic multiagent systems where agent interactions are described by game-theoretic principles.
The dynamics of such systems capture the agents' ability to learn from experience, allowing them to respond intelligently in settings not encountered during their training. 
 
\Gls{SIAR} is grounded upon polynomial regression and \gls{SOS} optimization \parencite{AMS_SC20_sum, lasserre_sum_2006} embedding it into a rich theoretical framework and allowing for formal approximation guarantees.
Critically, to have any hope of accurate system recovery given such data sparsity, our approach has to leverage side-information constraints.
By design, our framework allows the swift integration of a wide range of constraints, which are native to game-theoretic applications, such as positive correlation between the players' directions of motion and payoffs, game symmetries, a.o. 
We test our framework in a wide range of classic benchmarks including:
\begin{enumerate*}
  \item periodic replicator dynamics in matching pennies games;
  \item chaotic replicator dynamics in \param{\epsilon}-perturbed \gls{RPS} games; and
  \item equilibrium selection prediction in numerous game dynamics in multiplayer atomic congestion games with up to hundreds of pure, attracting Nash equilibria.
\end{enumerate*}
Thus, our methodology can also be used to provide insights into two  
fundamental problems of game theory: equilibrium selection and non-equilibrating (possibly chaotic) dynamics. 
In each of the above settings, our method showcases near perfect system recovery strongly outperforming previous benchmarks.   

\section{Related Work}   
\label{sec:RelatedWork}


System identification is the discipline that provides tools for learning the governing equations of a dynamical system using trajectory data \parencite{ljung_system_1999, isermann_identification_2011}. 
Considerable research has been dedicated to data-driven approaches for system identification, including deep learning \parencite{brunton_data-driven_2019, cranmer_discovering_2020}, symbolic regression \parencite{schmidt_distilling_2009,  udrescu_ai_2020a, udrescu_ai_2020b}, and statistical learning \parencite{lu_nonparametric_2019}. 
A prominent research direction, driven by the prevalence of sparse governing equations in physical systems, relies on identifying unknown systems through sparsity-promoting optimization techniques \parencite{brunton_discovering_2016, rudy_data-driven_2017, kaiser_sparse_2018}.
Nevertheless, all these approaches fall short in addressing time-critical scenarios, which are the focal point of consideration in this work.

The idea of integrating side-information to aid machine learning models in uncovering the governing equations of systems has been extensively explored, particularly within the research area of physics-informed learning \parencite{raissi_physics-informed_2019, chen_physics-informed_2021, liu_machine_2021}.
The use of \gls{SOS} optimization to search for regressors satisfying side-information constraints was initially introduced in the context of shape-constrained regression \parencite{curmei_shape-constrained_2023}.
Subsequently, \textcite{ahmadi_learning_2023} extended this approach to identify dynamical systems from limited observational data, employing side-information constraints to enhance the model's performance. 
However, the side-information constraints used by \textcite{ahmadi_learning_2023} are more tailored to physical systems, e.g. conservation laws, and not relevant to systems of strategic agents, which are the focus of this proposal.
In a game-theoretic setup properties arise as a consequence of the game, e.g., the anonymity of the players \parencite{brandt_symmetries_2009}, and characteristics of player behavior, e.g., risk-seeking/risk-averse tendency, bounded/unbounded rationality, a.o. \parencite{cranmer_discovering_2020}.

\section{Preliminaries}
\label{sec:Preliminaries}


A $\noplayers$\nobreakdash-player normal-form game $\game \eqdef (\noplayers, \pures, \payoff)$ denotes the interaction between a set of players $[\noplayers] \eqdef \set{1, \dots, \noplayers}$.
Each player~$i$ has a finite ordered set of strategies (or actions) $\pures_{i}$ of size $\nopures_{i}$, and a reward function~$\payoff_{i} \from \pures \to \reals$, where~$\pures \eqdef \pures_{1} \times \dots \times \pures_{\noplayers}$ denotes the set of all strategy profiles of~$\game$.
Typically, in a normal-form game, each player~$i$ is allowed to use mixed strategies $\action_{i} \eqdef (\action_{i 1}, \dots, \action_{i \nopures_{i}}) \in \actions_{i}$, where $\action_{i j}$ is the probability with which player~$i$ uses their $j$\nobreakdash-th strategy, and $\actions_{i}$ is the $(\nopures_{i} - 1)$\nobreakdash-simplex.
The players' reward functions naturally extend to the set of mixed strategy profiles~$\actions_{i} \eqdef \actions_{1} \times \dots \times \actions_{\noplayers}$ with 
\begin{equation}
  \payoff_{i}(\action) 
    \eqdef \mean_{\pure \sim \action}{\payoff_{i}(\pure)}
    = \sum_{j_{1} = 1}^{\nopures_{i}} \dots \sum_{j_{\noplayers} = 1}^{\nopures_{\noplayers}} \action_{1 j_{1}} \dotsm \action_{\noplayers j_{\noplayers}} \payoff_{i}(\pure_{1 j_{1}}, \dots, \pure_{\noplayers j_{\noplayers}})
      \qquad \action \in \actions, \,i \in [\noplayers].
\end{equation}

In game-theoretic setups, it is convenient to be able to distinguish between the strategy of some player~$i$ in a strategy profile~$\pure$ and the strategies of the rest of the players.
In this work, we use the standard game-theoretic shorthand $(\pure_{i}, \pure_{-i})$ to denote this distinction. 
Similarly, given some mixed strategy profile~$\action$, we write $(\action_{i}, \action_{-i})$ to distinguish between the mixed strategy of player~$i$ and the mixed strategies of the rest of the players.
Furthermore, if $\action_{i}$ represents some pure strategy~$\pure_{i j}$ of~$i$, we write $(\pure_{i j}, \action_{-i})$ to point on that fact.
Finally, we are going to use $\nash(\game)$ to denote the set of Nash equilibria of $\game$.
Recall that a strategy profile~$\action^{*}$ is called a Nash equilibrium if no player has an incentive to unilaterally deviate from it, i.e., 
\begin{equation}
  \payoff_{i}(\action^{*})
    \geq \payoff_{i}(\action_{i}, \action^{*}_{-i}) 
      \qquad \forall \action \in \actions, \,i \in [\noplayers].
\end{equation}


\subsection{Game Dynamics}
\label{sec:GameDynamics}
Game dynamics are dynamical systems that model the time-evolving behavior of the players under the assumption of repeated play.
Time-invariant game dynamics, i.e., dynamics that don't change as time progresses, are typically described by a system of ordinary differential equations of the form
\begin{subequations}
\label{eq:GameDynamics}
\begin{alignat}{2}
  \dot \action(t) 
    &= \policy\xpar[\big]{\action(t)} 
      &&\qquad \forall t \in \horizon \\
  \action(0) 
    &= \initialaction
      &&\qquad \initialaction \in \actions,
\end{alignat}
\end{subequations}
where the vector field $\policy \from \actions \to \reals^{\nopures_{1}} \times \dots \times \reals^{\nopures_{\noplayers}}$ denotes the players' update policies.
The fundamental notion of every dynamical system is that of a flow function $\flow \from \horizon \times \actions \to \reals^{\nopures_{1}} \times \dots \times \reals^{\nopures_{\noplayers}}$ that gives us access to solutions of the initial value problem in~\eqref{eq:GameDynamics}.
Specifically, for any initial strategy profile~$\initialaction$ and any time~$t$, we define the flow $\flow(t, \initialaction)$ to be the initial value solution $\action(t)$ in~\eqref{eq:GameDynamics} given that $\action(0) = \initialaction$.
The examples that follow are centered around fundamental game dynamics such as the \param{q}-replicator dynamics \parencite{giannou_rate_2021} that include the well-known (Euclidean) projection dynamics, replicator dynamics, and log-barrier dynamics as special cases, as well as the smooth Q\nobreakdash-learning dynamics \parencite{tuyls_evolutionary_2006}.
The definitions of the all aforementioned dynamics can be found in \cref{app:GameDynamics}.
A list of abbreviations can be found in \cref{app:Abbreviations}.

\section{The \glsentrylong{SIAR} Problem}
\label{sec:SIARProblem}


In this work, we present a computational framework for discovering the game dynamics that govern players' strategic evolution in a $\noplayers$\nobreakdash-player normal-form game $\game$, which we assume it's known.
Our framework, which we call \glsxtrfull{SIAR}, builds on recent advances in shape-constrained regression and \gls{SOS} optimization \parencite{curmei_shape-constrained_2023, ahmadi_learning_2023}, can operate in real-time and requires access to only limited observations of the system.

The goal of the \gls{SIAR} framework is to identify the true vector field $\policy \from \actions \to \reals^{\nopures_{1}} \times \dots \times \reals^{\nopures_{\noplayers}}$ that describes the players' behavior (or to approximate it, in case it is not polynomial) using a constant number~$\nodata$ of---possibly noisy---observations $\set[\big]{\action(t_{1}), \dots, \action(t_{\nodata})}$, where $t_{1} < \dots < t_{\nodata}$, from a short run of a single system trajectory for some initialization~$\initial \in \actions$, and estimations of their corresponding velocities, which we denote by $\dot \action(t_{1})$, \dots, $\dot \action(t_{\nodata})$. 
The duration of the short-run is given by $\Delta T \eqdef t_{\nodata} - t_{1}$.
We assume that the limited number of data~$\nodata$ and the duration of the short-run~$\Delta T$ are significant prohibitive factors, and the use of unconstrained regression to solve the problem cannot provide a solution with sufficient accuracy.

To overcome the challenges incurred by the absence of an adequate number of data, we incorporate information into the regression problem, which for reference we name the \gls{SIAR} problem, in the form of side-information constraints.
These constraints capture fundamental properties, assumptions, and expectations about player behavior and their space of strategies, and are enforced computationally using \gls{SOS} optimization.
Specifically, we assume that we have access to information about the policies~$\policy$ that can be encoded as polynomial side-information constraints over subsets of $\reals^{\nopures_{1}} \times \dots \times \reals^{\nopures_{\noplayers}} \times \actions$ in the form
\begin{subequations}
\label{eq:SideInformationConstraints}
\begin{alignat}{2}
  \ngcon_{\ell}\xpar[\big]{\policy(\action), \action}
    &\geq 0
      &&\qquad \forall \action \in \ngcons_{\ell} \subseteq \actions \\
  \eqcon_{\tau}\xpar[\big]{\policy(\action), \action} 
    &= 0
      &&\qquad \forall \action \in \eqcons_{\tau} \subseteq \actions,
\end{alignat}
\end{subequations}
and, for all $\ell$ and $\tau$, the following two statements hold:
\begin{enumerate}
    \item The subsets $\ngcons_{\ell}$ and $\eqcons_{\tau}$ are basic semialgebraic sets.
    \item The functions $\ngcon_{\ell}$ and $\eqcon_{\tau}$ are polynomials.
\end{enumerate}
In later \lcnamecrefs{sec:SIARProblem}, we are going to present numerous such properties and their formulations as side-information constraints.

For concreteness, we now introduce the complete formulation of a generic \gls{SIAR} problem that minimizes the mean squared error with respect to the system's observed data and approximates the update policies~$\policy_{i}$ of each player~$i$ by a polynomial vector field~$p_{i}$:
\begin{equation}
\label{eq:SIAR}
\begin{aligned} 
  \underset{p_{1}, \dots, p_{\noplayers}}{\text{minimize}} 
    &&&\sum_{i = 1}^{\noplayers} \sum_{k = 1}^{\nodata} \norm{p_{i}\xpar[\big]{\action(t_{k})} - \dot \action_{i}(t_{k})}^{2} \\
  \text{subject to}
    &&&\begin{alignedat}[t]{2}
      &p_{i} \in \polys^{\nopures_{i}}[\action]
        &&\qquad i = 1, \dots, \noplayers \\
      &\ngcon_{\ell}\xpar[\big]{p(\action), \action} \geq 0
        &&\qquad \forall \action \in \ngcons_{\ell} \\
      &\eqcon_{\tau}\xpar[\big]{p(\action), \action} = 0
        &&\qquad \forall \action \in \eqcons_{\tau},
    \end{alignedat}
\end{aligned} \tag{SIAR}
\end{equation}
where, for each player~$i$, $\polys^{\nopures_{i}}[\action]$ denotes the set of polynomial vector fields of size~$\nopures_{i}$ and indeterminates~$\action$.
In general, the \gls{SIAR} problem is computationally intractable.
However, it is possible to use the \gls{SOS} optimization framework to obtain a hierarchy of increasingly better solutions to the original problem.

\section{Side-Information Constraints}
\label{sec:SideInformationConstraints}


In this \lcnamecref{sec:SideInformationConstraints}, we explore some different types of information that are relevant to strategic multiplayer environments, and we model them as side-information constraints.
We observe that, in general, such information can be classified into two types:
\begin{enumerate*}
  \item information that captures geometric properties of the game dynamics, such as the invariance of the mixed strategy space of the game; and
  \item information that expresses assumptions or desiderata about the players' behavior, such as their desire to explore their strategy space as opposed to harnessing immediate rewards.
\end{enumerate*}
In the following \lcnamecrefs{sec:NormalFormGameProperties}, we describe and model as side-information constraints several game-theoretic properties that fall in these two types.
Interestingly replicator dynamics satisfy most of these properties; see \cref{app:ReplicatorDynamicsProperties}.


\subsection{Properties of Strategic Interactions in Normal-Form Games}
\label{sec:NormalFormGameProperties}


\subsubsection{Forward-Invariance of the Strategy Space}
\label{sec:ForwardInvariance}
Concerning, the former type, maybe the most natural side-information constraints are the ones modeling the forward (or positive) invariance, of the mixed strategy space~$\actions$.
Recall that a subset $\actions' \subseteq \actions$ is forward-invariant if any trajectory that is initialized in $\actions'$, i.e., $\initial \in \actions'$, is contained entirely in $\actions'$ for all positive times, i.e., $\flow(t, \initial) \in \actions'$ for all $t \geq 0$.
In terms of the dynamics' update policies~$\policy$, the above means that for each mixed strategy profile~$\action \in \actions$, and for each player~$i$, the direction of motion of $\action_{i}$ given by $\policy_{i}(\action)$ is a direction that player~$i$ can update, infinitesimally, their strategy without moving outside their mixed strategy space~$\actions_{i}$. 
Formally, $\policy(\action)$ lies in the tangent cone $\tangent_{\actions}(\action)$ of $\actions$ at the point~$\action$ \parencite{nagumo_uber_1942}, where the tangent cone $\tangent_{\actions}$ is given by 
\begin{equation}
  \tangent_{\actions}(\action)
    \eqdef \set*{
      y \in \bigtimes_{i = 1}^{\noplayers} \reals^{\nopures_{i}} 
      \given \begin{alignedat}{2}
        \smashoperator{\sum_{j = 1}^{\nopures_{1}}} y_{i j}
          &= 0
            &&\qquad \forall i \in [\noplayers] \\
        y_{i j}
          &\geq 0
            &&\qquad\text{whenever $\action_{i j} = 0$}
      \end{alignedat}
    }
      \qquad \action \in \actions.
\end{equation}
At this point, we remark that the sets $\actions$ and $\set{\action \in \actions \given \action_{i j} = 0}$ are basic semialgebraic sets, and therefore we can construct a set of side-information constraints as given in~\eqref{eq:SideInformationConstraints} that incorporate the forward-invariance property of~$\actions$.
Specifically, for each player~$i$ it suffices to constrain~$\policy_{i}$ such that:
\begin{subequations}
\label{eq:StrategySpaceForwardInvariance}
\begin{alignat}{3}
  \sum_{j = 1}^{\nopures_i} \policy_{i j}(\action)
    &= 0
      &&\qquad \forall \action \in \actions \\
  \policy_{i j}(\action)
    &\geq 0
    &&\qquad\text{$\forall \action \in \actions$ such that $\action_{i j} = 0$}
      &&\qquad j \in [\nopures_{i}].
\end{alignat}
\end{subequations}



\subsubsection{Game Symmetries}
\label{sec:GameSymmetries}
Side-information constraints can also be used to model the fact that symmetries of the underlying normal-form game, which in layman's terms imply that players or strategies are in some way interchangeable \parencite{brandt_symmetries_2009}, should also be reflected in the game dynamics.  
As a concrete example, consider a game among anonymous players---a type of game where the players' utility functions are invariant under any permutations of the players; the analysis of anonymous games can be traced back at least to the seminal work of \textcite{nash_non-cooperative_1951}.
Since in such a game the players have no information about the identities of their opponents, the evolution of their strategies must not be based on that information either.

Formally, we say that a normal-form game~$\game$ consists of anonymous players if:
\begin{enumerate*}
  \item the players have identical strategies spaces, i.e., $\actions_{1} = \dots = \actions_{\noplayers}$; and
  \item for all permutation functions $\pi \in \perms(\noplayers)$ of the players, it holds
\end{enumerate*}
\begin{equation}
\label{eq:GameAnonymity}
  \payoff_{\pi(i)}(\pure) 
    = \payoff_{i}\xpar[\big]{\pi(\pure)}
  \quad\text{where}\quad
  \pi(\pure)_{i} 
    \eqdef \pure_{\pi(i)}
      \qquad \forall \pure \in \pures, \,i \in [\noplayers].
\end{equation}
If~\eqref{eq:GameAnonymity} seems counter-intuitive at a first sight, in \cref{app:PlayerAnonymity} we have an example to build the necessary intuition.
Nonetheless, having established the above game-theoretic property, one would expect that if a normal-form game~$\game$ has anonymous players, the anonymity would be reflected in the evolution of the players' strategies.
Indeed, in many game-theoretic setups it is natural to assume that, for all permutations $\pi \in \perms(\noplayers)$, it holds
\begin{equation}
\label{eq:UpdatePolicyAnonymity}
  \policy_{\pi(i)}(\action) 
    = \policy_{i}\xpar[\big]{\pi(\action)}
  \quad\text{where}\quad
  \pi(\action)_{i} 
    \eqdef \action_{\pi(i)}
      \qquad \forall \action \in \actions, \,i \in [\noplayers].
\end{equation}
Notice that the above can directly be imposed as side-information constraints.
However, in doing so, the number of introduced constraints to the \gls{SIAR}problem would be exponential in~$\noplayers$ due to the dependence of \eqref{eq:UpdatePolicyAnonymity} to $\abs{\perms(\noplayers)}$.
In practice, the above equivalence can be used instead to simplify the \gls{SIAR} problem.
Specifically, since~\eqref{eq:UpdatePolicyAnonymity} holds for any permutation~$\pi$, one can always set $\pi_{i}(\pure) = i + 1 \mod_{1} \noplayers$, and recover the update policies $\policy_{2}$ \dots, $\policy_{\noplayers}$ from $\policy_{1}$.
Therefore, assuming the players are anonymous, we can restrict the search space of the \gls{SIAR} problem from $\polys^{\nopures_{1}}[\action] \times \dots \times \polys^{\nopures_{\noplayers}}[\action]$ to simply $\polys^{\nopures_{1}}[\action]$.


\subsubsection{Existence of Uncoupled Formulation} 
\label{sec:UncoupledFormulation}
Recall that if the game dynamics are uncoupled the evolution of the strategies of each player~$i$ depend on the strategies~$\action_{-i}$ of the other players only through their own reward function $\payoff_{i}$, i.e., given that the vector of instantaneous rewards of $i$ at $\action$, $\instantpayoff_{i}(\action)$, defined as 
\begin{equation}
\label{eq:InstantaneousPayoffs}
  \instantpayoff_{i j}(\action) 
    \eqdef \payoff_{i}(\pure_{i j})
      \qquad j \in [\nopures_{i}],
\end{equation}
remains constant, the update policy~$\policy_{i}$ is invariant under changes of $\action_{-i}$ \parencite{hart_uncoupled_2003}.
Specifically, the game dynamics of a normal-form game~$\game$ are said to be uncoupled if, for each player~$i$, there exists some vector field~$\uncoupled_{i} \from \actions_{i} \times \reals^{\nopures_{i}} \to \reals^{\nopures_{i}}$ such that
\begin{equation}
\label{eq:UncoupledUpdatePolicy}
  \policy_{i}(\action) 
    = \uncoupled_{i}\xpar[\big]{\action_{i}, \instantpayoff_{i}(\action)}
      \qquad \forall \action \in \actions.
\end{equation}
Notice that~\eqref{eq:UncoupledUpdatePolicy}, as is the case with~\eqref{eq:UpdatePolicyAnonymity}, can be directly imposed as side-information constraints.
Moreover, from a computational perspective the existence of an uncoupled formulation is a powerful property that it can reduce the search of each update policy $\policy_{i} \from \actions \to \reals^{\nopures_{i}}$ of player~$i$ to the search of a vector field $\uncoupled_{i} \from \actions_{i} \times \reals^{\nopures_{i}} \to \reals^{\nopures_{i}}$ that depends on a potentially smaller number of indeterminates.
The uncoupling of the game dynamics is, often, assumed to aid in their analysis, and is actually satisfied by many well-known dynamics such the \param{q}-replicator dynamics and the Q\nobreakdash-learning dynamics.

\subsection{Common Assumptions for Strategic Interactions}
\label{sec:NormalFormGameAssumptions}


\subsubsection{Nash Stationarity}
\label{sec:NashStationarity}
In addition to the properties presented in the previous \lcnamecref{sec:NormalFormGameProperties}, side-information constraints can also model various behavioral traits of the players.
One important such trait, exhibited by various game dynamics, is the Nash stationarity.
This property states that any trajectory initialized in a Nash equilibrium of the game remains to that point for all positive times, i.e., the Nash equilibria are a subset of the rest-points of the system.
Formally, given some normal-form game~$\game$, we say that the dynamics satisfy the Nash stationarity if
\begin{equation}
\label{eq:NashStationarity}
  \policy(\action^{*}) 
    = 0
      \qquad \forall \action^{*} \in \nash(\game).
\end{equation}
%
Notice that not all game dynamics satisfy this property.
For example, the fixed-points of the Q\nobreakdash-learning dynamics (cf.~\cref{sec:QLearningDynamics}) correspond to \glspl{QRE}, which coincide with the Nash equilibria of the game if, and only if, the explore rate~$\explorerate$ of the Q\nobreakdash-learning dynamics is zero \parencite{leonardos_exploration-exploitation_2021}.


\subsubsection{Positive Correlation} 
\label{subsubsec:PositiveCorrelation}
The strategic nature of the players can, often, be expressed by the relationship between their directions of motion~$\dot \action(t)$ and their reward functions~$\payoff\xpar[\big]{\action(t)}$.
One property that features this relationship is the positive correlation property, which expresses the adversity of each player~$i$ of decreasing its reward by unilaterally deviating from $\action(t)$ \parencite{swinkels_adjustment_1993, sandholm_population_2010}.
Formally, the positive correlation property is given, for each player~$i$, by the constraints:
\begin{equation}
\label{eq:PositiveCorrelation}
  \inner[\big]{\policy_i(\action)}{\instantpayoff_{i}(\action)} 
    > 0
      \qquad\text{$ \forall \action \in \actions$ such that $\policy(\action) \neq 0$}.
\end{equation}
The exact relationship between $\dot \action_{i}(t)$ and $\payoff_{i}\xpar[\big]{\action(t)}$ can be recovered from the above by noticing that, for all $j \in \set{1, \dots, \nopures_{i}}$, and mixed action profiles~$\action$, we have
\begin{equation*}
  \frac{\partial \payoff_{i}(\action)}{\partial \action_{i j}}
    = \frac{\partial}{\partial \action_{i j}} \xpar[\Bigg]{\sum_{k = 1}^{\nopures_{i}} \action_{i k}\payoff_{i}(\pure_{i k}, \action_{-i})}
    = \payoff_{i}(\pure_{i j}, \action_{-i})
    = \instantpayoff_{i j}(\action),
\end{equation*}
and subsequently, the left-hand side of~\eqref{eq:PositiveCorrelation} is equivalent to
\begin{equation*}
  \inner[\big]{\policy_{i}\xpar[\big]{\action(t)}}{\instantpayoff_{i}\xpar[\big]{\action(t)}}
    = \inner[\big]{\policy_{i}\xpar[\big]{\action(t)}}{\gradient_{\action_{i}} \payoff_{i}\xpar[\big]{\action(t)}}
    = \inner[\big]{\dot \action_{i}(t)}{\gradient_{\action_{i}} \payoff_{i}\xpar[\big]{\action(t)}}.
\end{equation*}

In the sequel, we will keep referring to~\eqref{eq:PositiveCorrelation} as the positive correlation property.
However, one should keep in mind that these inequalities only capture the positive correlation of the directions of motion $\dot \action(t)$ at time $t$ with a rather special---nonetheless important---quantity $\instantpayoff\xpar[\big]{\action(t)}$.
As a possible generalization of the above, one could express the positive correlation of $\action(t)$ and the parametrized vectors $\instantpayoff'_{i} \from \actions \times \nnreals \to \reals^{\nopures_{i}}$, for $i \in [\noplayers]$, given by
\begin{equation}
\label{eq:ParametrizedInstantenuousPayoffs}
  \instantpayoff'_{i j}(\action; \explorerate) 
    = \instantpayoff_{i j}(\action) - \explorerate (\ln \action_{i j} + 1)
    \qquad\text{for all $\action \in \actions$}
      \qquad j \in [\nopures_{i}],
\end{equation}
where $\explorerate \in \nnreals$ is some nonnegative constant.
Following a similar reasoning as before, we can show that the above property is captured, for each player~$i$, by the constraints:
\begin{equation}
\label{eq:GeneralizedPositiveCorrelation}
  \inner[\big]{\policy_{i}(\action)}{\instantpayoff'_{i}(\action; \explorerate)} 
    > 0
      \qquad\text{for all $\action \in \actions$ such that $\policy(\action) \neq 0$}.
\end{equation}
A positive correlation between the directions of motion $\action(t)$ and $\instantpayoff'\xpar[\big]{\action(t); \explorerate}$ expresses the adversity of each player~$i$ of decreasing their reward by a unilaterally deviating from $\action$ but, at the same time factors in, at weight $\explorerate$, the uncertainty of that deviation.
This property is the one precisely satisfied by the Q\nobreakdash-learning dynamics with explore rate $\explorerate$.

\begin{remark}
\label{rem:ConstraintsRelaxation}
  Technically, \cref{eq:PositiveCorrelation,eq:GeneralizedPositiveCorrelation} cannot be imposed as side-information constraints, since they feature nonnegativity constraints that are required to be strict for some subset of the strategy space.
  In practice, we relax those and similar conditions asking only for their nonnegativity in the whole strategy space. 
  The rationale behind these relaxations is that if the unknown game dynamics satisfy one of those properties, then they necessarily satisfy their relaxations, and therefore the unknown policy lies in the solution set of the \gls{SIAR} problem where the relaxations are imposed.
\end{remark}


\subsubsection{Elimination of Strictly Dominated Strategies}
\label{sec:DominatedStrategies}
We conclude this \lcnamecref{sec:NormalFormGameProperties}, with yet another natural behavioral trait of the players; their tendency to eventually disregard strategies that are strictly dominated by alternatives.
Let us recall that a pure strategy $\pure_{i} \in \pures_{i}$ of player~$i$ is strictly dominated if for any mixed strategy profile $\action_{-i} \in \actions_{-i}$ of the rest of the players, there exists some mixed strategy~$\action'_{i}$ of $i$ such that $\payoff_{i}(\action'_{i}, \action_{-i}) > \payoff_{i}(\pure_i, \action_{-i})$.
In repeated play, it is reasonable to expect the players to become aware of the above imbalance in their rewards and gradually lean away from such strategies of poor-performance.
Formally, the elimination of strictly dominated pure strategies of player~$i$ is given by
\begin{equation}
\label{eq:DominatedActionElimination}
  \lim_{t \to +\infty} \flow_{i j}(t, \initial)
    = 0
      \qquad\text{whenever $\pure_{i j}$ is strictly dominated}
      \qquad j \in [\nopures_{i}].
\end{equation}

Such eliminations could be progressively generalized by expanding on common knowledge assumptions.
Specifically, suppose the eventual elimination of the strictly dominated strategies $\pures' \subseteq \pures$ is common knowledge.
In that case, the players, as strategic entities, can act on that by eventually eliminating strictly dominated strategies of all the restrictions of $\game$ where some subset of $\pures'$ is unavailable. 
Now, if the latter is also common knowledge, then the same reasoning can be applied once again for a potentially extended set~$\pures'$ that includes strategies that are strictly dominated in the aforementioned restrictions.
In fact, the same reasoning can be repeatedly applied until the set~$\pures'$ cannot be extended further.
For reference, if a pure strategy~$\pure_{i}$ of $i$ belongs to the set $\pures'_{i}$ we say that it is strictly dominated by recursion. 

Although, in general, the above properties cannot be modeled as side-information constraints, there exists a set of sufficient conditions, identified in the significant work of \textcite{hofbauer_evolutionary_1996}, that lead to them.
In particular, the result requires that the dynamics belong to the special class of regular selection dynamics and satisfy a property known as convex monotonicity.
The dynamics are said to belong to the class of regular selection dynamics if, for all players~$i$, the update policies~$\policy$ take the form 
\begin{subequations}
\label{eq:RegularSelectionDynamicsUpdatePolicy}
\begin{alignat}{3}
  \sum_{j = 1}^{\nopures_{i}} \action_{i j} \cdot g_{i j}(\action) 
    &= 0
      &&\quad \forall \action \in \actions \label{subeq:RSDForwardInvariance} \\
  \action_{i j} \cdot g_{i j}(\action) 
    &= \policy_{i j}(\action)
      &&\quad \forall \action \in \actions
      &&\qquad j \in [\nopures_{i}]
\end{alignat}
\end{subequations}
where the functions $g_{i j} \from \actions \to \reals$ have an open domain and are locally Lipschitz continuous in $\actions$.
Note that, if $g_{i j}$ are polynomials, then, since $\actions$ is compact, the latter conditions are clearly met.
Furthermore, the regular selection dynamics are said to satisfy the convex monotonicity property, if for all players~$i$ and all $j \in [\nopures_{i}]$ such that $\pure_{i j} \in \pures$ is strictly dominated, possibly by recursion, it holds 
\begin{equation}
\label{eq:ConvexMonotonicity}
    \inner[\big]{\action'_{i}}{g_{i}(\action)} 
      > g_{i j}(\action)
        \qquad\text{$\forall \action'_{i} \in \actions_{i}$, and $\action \in \actions$ such that $\payoff_{i}(\action'_{i}, \action_{-i}) > \payoff_{i}(\pure_{i j}, \action_{-i})$}.
\end{equation}
Notice that the set $\set[\big]{(\action'_{i}, \action) \in \actions_{i} \times \actions \given \payoff_{i}(\action'_{i}, \action_{-i}) > \payoff_{i}(\pure_{i j}, \action_{-i})}$ is semialgebraic, and therefore we can search for appropriate update policies~$\policy$ and vector fields~$g$ that satisfy the property of the elimination of strictly dominated strategies by relaxing (cf.~\cref{rem:ConstraintsRelaxation}) and imposing \cref{eq:RegularSelectionDynamicsUpdatePolicy,eq:ConvexMonotonicity} as side-information constraints to the \gls{SIAR} problem.

\section{Solving the \glsentryshort{SIAR} Problem}
\label{sec:SIARSolution}


The list of side-information constraints presented in the previous section is, by no means, an exhaustive one.
However, they model a significant subset of contextual information that might be available to a mechanism designer seeking to understand the behavior of the players.
On that note, we are now ready to explain how such a mechanism designer can solve a \gls{SIAR} problem, and therefore discover this behavior.
To achieve the above, we rely on a cornerstone result in modern convex optimization, which allows us to efficiently search for nonnegative polynomials over a basic semialgebraic set \parencite{parrilo_structured_2000, prestel_positive_2001, lasserre_global_2001, parrilo_semidefinite_2003, lasserre_sum_2006, laurent_sums_2009} using \gls{SOS} optimization.
Specifically, let
\begin{equation}
  \semiset 
    = \set*{
      x \in \reals^{n} 
      \given \begin{alignedat}{2}
        \ngcon_\ell(x) &\geq 0
          &&\qquad \ell = 1, \dots, m_{\ngcon} \\
        \eqcon_\tau(x) &= 0
          &&\qquad \tau = 1, \dots, m_{\eqcon}
      \end{alignedat}
    }
\end{equation}
be some semialgebraic set, where $n, m_{\ngcon}, m_{\eqcon} \in \naturals$ are some arbitrary constants.
The result states that we can search, using semidefinite programming, for a fixed-degree~$d_{p}$ polynomial $p \from \reals^{n} \to \reals$, as well as some fixed-degree~$d_{q_{\tau}}$ polynomials $q \from \reals^{n} \to \reals$, for $\tau = 1$, \dots, $m_{\eqcon}$ and fixed-degree~$2d_{\sos_{\ell}}$ \gls{SOS} polynomials $\sos_{\ell} \from \reals^{n} \to \nnreals$, for $\ell = 0$, \dots, $m_{\ngcon}$, that satisfy the Putinar-type decomposition
\begin{equation}
    p(x) 
      = \sos_{0}(x) + \sum_{\tau = 1}^{m_{\ngcon}} q(x) \ngcon_{\tau}(x) + \sum_{\ell = 1}^{m_{\eqcon}} \sos(x) \eqcon_{\ell}(x).
\end{equation}
Clearly, the existence of such a decomposition is a sufficient condition for the nonnegativity of the polynomial $p$ over the semialgebraic set $\semiset$.
The condition is also necessary under additional mild assumptions described by \textcite[Theorem 3.20]{laurent_sums_2009}.
Furthermore, by increasing the degrees $d_{p}$, $d_{q_{\ell}}$ and $d_{\sos_{\tau}}$ we can create a hierarchy of semidefinite programs that provide better approximations to $p$; we refer the interested reader to \cref{app:SIARExample} for an example.

\section{Approximation Guarantees of the \glsentryshort{SIAR} Framework}
\label{sec:ApproximationGuarantees}


As we demonstrated in the previous \lcnamecref{sec:SIARSolution}, for every fixed-degree~$d$ we get a distinct \gls{SOS} problem resulting in a hierarchy of increasingly better approximations to the value of the \gls{SIAR} problem.
However, it is unclear whether this approach can consistently provide good approximations to the update policies~$\policy$ for two reasons.
Firstly, we are attempting to approximate a potentially non-polynomial vector field~$\policy$ using a polynomial one.
Secondly, since we are strengthening the polynomial nonnegativity constraints expressing the side-information constraints using Putinar-type certificates, even if $\policy$ is indeed a polynomial vector field, we have no guarantee that it satisfies those certificates.

Despite these challenges, a recent result by \textcite[Theorem 5.5]{ahmadi_learning_2023} provides to the hierarchy of \gls{SOS} relaxations important approximation guarantees given that the side-information constraints imposed in the \gls{SIAR} problem satisfy certain conditions.
In \cref{thm:SideInformationConstraintsMetrics}, whose proof can be found in \cref{app:ApproximationGuaranteesProofs}, we show that these conditions are indeed satisfied by all side-information constraints we presented in \cref{sec:SideInformationConstraints}.
\begin{restatable}{lemma}{SideInformationConstraintsMetrics}
\label{thm:SideInformationConstraintsMetrics}
  For each type of the side-information constraints~$S$ in \cref{sec:SideInformationConstraints}, we can construct functionals $L_{S}(f) \from \condiffs(\actions) \to \reals$ that quantify how close some vector field of update policies $f$ satisfies $S$.
  In particular, $L_{S}$ satisfy the following properties:
  \begin{enumerate}
    \item $L_{S}(f) = 0$ if, and only if, $f$ satisfies $S$.
    \item For any $\delta > 0$, there exists $\epsilon > 0$, such that $\abs{L_S(f) - L_S(f')} \leq \delta$ whenever $f, g \in \condiffs(\actions)$ satisfy:
    \begin{equation}
        \max_{\action \in \actions} \norm{f(\action) - g(\action)} 
          \leq \epsilon 
        \quad \text{and} \quad  
        \max_{\action, i, j, k, \ell} \abs[\bigg]{\frac{\partial f_{i j}(\action)}{\partial \action_{k \ell}} - \frac{\partial g_{i j}(\action)}{\partial \action_{k \ell}}}  
          \leq \epsilon,
    \end{equation}
    where $\action \in \actions$, $i, k \in [\noplayers]$, $j \in [\nopures_{i}]$, and $j \in [\nopures_{k}]$.
  \end{enumerate}
\end{restatable}

By the above \lcnamecref{thm:SideInformationConstraintsMetrics}, and using Theorem 5.5 of \textcite{ahmadi_learning_2023}, we have the following \lcnamecref{thm:ApproximationGuarantees}:
\begin{corollary}
\label{thm:ApproximationGuarantees}
  Consider a normal-form game~$\game$ whose players update their strategies based on some update policies $\policy \from \actions \to \reals^{\nopures_{1}} \times \dots \times \reals^{\nopures_{\noplayers}}$.
  If $\policy \in \condiffs(\actions)$ is satisfying any combination $S_{1}$, \dots, $S_{K}$ of the properties in \cref{sec:SideInformationConstraints}, then for every time horizon~$T$ and approximation error $\epsilon > 0$ there exists some polynomial vector field $p \from \actions \to \reals^{\nopures_{1}} \times \dots \times \reals^{\nopures_{\noplayers}}$ such that
  \begin{subequations}
  \begin{alignat}{2}
    \max_{\action \in \actions} \norm{\flow_{\policy}(t, \initial) - \flow_{p}(t, \initial)}
      &\leq \epsilon 
        &&\qquad \forall \initial \in \actions, \,t \in [0, T] \\
      L_{S_{k}}(p)
        &\leq \epsilon
          &&\qquad \forall k \in [K],
  \end{alignat}
  \end{subequations}
  where $\flow_{\policy}$, $\flow_{p} \from \horizon \times \actions \to \actions$ are the flow functions corresponding to the update polices~$\policy$ and~$p$, respectively, and $L_{S_{1}}$, \dots, $L_{S_{K}}$ are the metrics that corresponds to the properties $S_{1}$, \dots, $S_{K}$ as these are defined in \cref{thm:SideInformationConstraintsMetrics}.
  Furthermore, the nonnegativity conditions of the side-information constraints of $S_{1}$, \dots, $S_{K}$ with respect to the update policies~$p$ can be certified by Putinar-type \gls{SOS} certificates.
\end{corollary}
In other words, \cref{thm:ApproximationGuarantees} guarantees that, for any finite time horizon and any desirable accuracy, there exists a degree of the \gls{SOS} hierarchy of the \gls{SIAR} problem whose solutions can approximate to the desired accuracy and for the given time-horizon, the trajectories of the unknown game dynamics given by $\policy$, and approximately satisfy any properties in \cref{sec:SideInformationConstraints} that we imposed to the optimization problem.
Although, this \lcnamecref{thm:ApproximationGuarantees} does not indicate the required degree that achieves those guarantees, in practice achievable $d$ can be relatively small;
see \cref{app:SIARExample} for an example of a \gls{SIAR} solution.

\section{Experimental Evaluation of the \glsentryshort{SIAR} Framework}
\label{sec:Experiments}


In this \lcnamecref{sec:Experiments}, we perform an empirical evaluation of the \gls{SIAR} framework in various normal-form games, game dynamics---polynomial and non-polynomial---and significant game-theoretical benchmarks such as the applicability of the \gls{SIAR} solution in equilibrium selection problems, and the recovery of chaotic dynamical systems.
These benchmarks were selected to demonstrate the potential of the \gls{SIAR} framework and comparative performance to state-of-the-art alternatives like the \gls{SINDy} framework by \textcite{brunton_data-driven_2019}.


\subsection{Comparative Performance of the \glsentryshort{SIAR} Framework}
\label{sec:SIARVsSINDy}
In this \lcnamecref{sec:SIARVsSINDy}, we compare the performance of the \gls{SIAR} framework with the one of the \gls{SINDy} framework introduced by \textcite{brunton_data-driven_2019} in recovering the update policies~$\policy$ of the replicator dynamics for the matching pennies game we analyzed in \cref{ex:SIARProblem,ex:SIARProblemSolution}.
Specifically, we consider the trajectory~$\set{\flow(t, \initialaction) \given t \in \horizon}$, where $\initialaction = \xpar[\big]{(0.25, 0.75), (0.125, 0.875)}$, and we use the dataset $\action(t_{k}) = \flow(t_{k}, \initialaction)$ and $\dot \action(t_{k}) = \policy\xpar[\big]{\action(t_{k})}$, where $t_{1} = 0.0$, $t_{2} = 0.2$, \dots, $t_{5} = 0.8$; initially, no noise is added to dataset.
Then, fixing degree~$d = 3$ for the solutions, we approximate~$\policy$ by, both, a solution~$\sindysol$ to the \gls{SINDy} problem and a solution~$\siarsol$ to the \gls{SIAR} problem, further imposing the forward invariance of the state space and the positive correlation property as side-information constraints for the latter.
Finally, we evaluate the performance of $\sindysol$ and $\siarsol$ on task of predicting the system trajectory in the time-interval $[t_{*}, 10]$, where $t_{*} = 1$ denotes the evaluation time.

As we depict in \cref{fig:SIARvsSINDy}, the accuracy of the \gls{SIAR} solution~$\siarsol$ far exceeds the one of the \gls{SINDy} solution~$\sindysol$ in this task.
\begin{figure}[hb!]
  \centering
  \includegraphics[width=\textwidth]{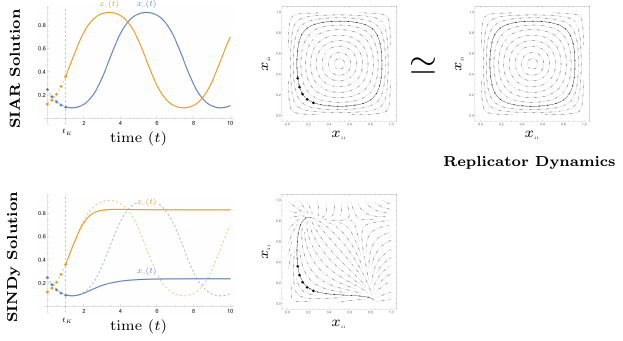}
  \caption{%
    On the left, are the trajectories of the dynamical system as predicted by the \gls{SIAR} and \gls{SINDy} frameworks in the interval $[t_{*}, 10]$, where $t_{*} = 1$ denotes the evaluation time.
    The true trajectory, depicted in faint dashed lines, is given by the replicator dynamics for the matching pennies game initialized at $\xpar[\big]{(0.25, 0.75), (0.125, 0.875)}$. 
    The two solutions were computed using a sample of five data points at times $t = 0.0, 0.2, \dots, 0.8$, depicted as dots.
    For the \gls{SIAR} problem, we used the forward invariance of the strategy space and the positive correlation property as side-information constraints.
    The \gls{SIAR} solution far outperforms the \gls{SINDy} solution, achieving near perfect recovery of the true trajectory.
    On the right, are the phase plots of the two solutions, where the aforementioned trajectory is depicted in black, and the dataset as black dots.
    The phase plot of the \gls{SIAR} solution is nearly identical to the true phase plot of the system, while the phase plot of the \gls{SINDy} solution is accurate only around the given dataset.
  }
  \label{fig:SIARvsSINDy}
\end{figure} 
Moreover, $\sindysol$ is accurate only around the evaluation point $t_{*}$, i.e., in the proximity of the given dataset.
In other words, $\sindysol$ acts as a local approximation to~$\policy$ around $t_{*}$.
On the other hand, $\siarsol$ provides a reliable approximation of the system trajectory for all times in $[t_{*}, 10]$, and, in fact, its phase plot with respect to the $\action_{1 1}$ and $\action_{1 2}$ is identical to the one of the true system, and therefore, $\siarsol$ provides a global approximation.
Adding noise to the dataset, the performance of the \gls{SIAR} solution when different sets of side-information are imposed is depicted in \cref{fig:UsingSIARWithNoisyData}.
\begin{figure}[hb!]
  \centering
  \includegraphics[width=\textwidth]{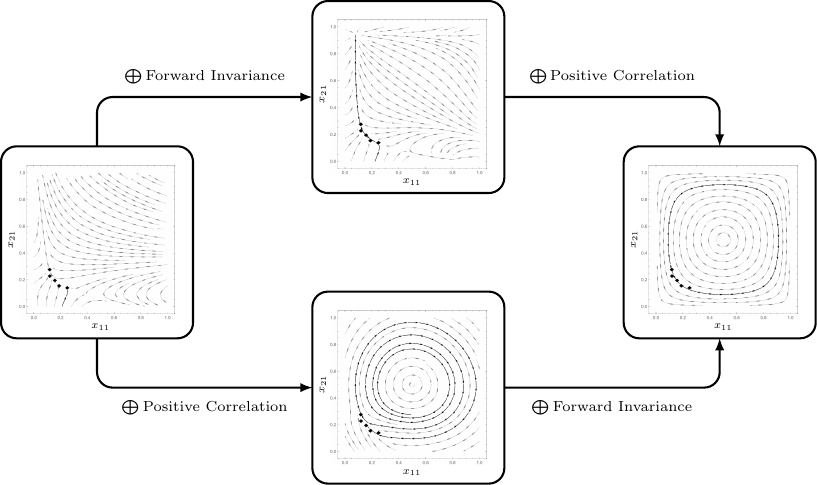}%
  \caption{%
    Phase plots of the \gls{SIAR} solutions for different sets of side-information constraints on the problem of identifying the replicator dynamics for the matching pennies game given a noisy sample of five data points, depicted as red dots, from a single trajectory of the system depicted in black.
    As we add side-information constraints, the accuracy of the solutions increase, and near perfect recovery is achieved whenever, both, the forward invariance of the strategy space and the positive correlation property are imposed to the \gls{SIAR} problem.
  }
  \label{fig:UsingSIARWithNoisyData}
\end{figure}
These solutions were computed using a noisy dataset of $K = 5$ data points $\tilde \action(t_{k}) \sim \normal\xpar[\big]{\action(t_{k}), 0.01}$ and $\dot \tilde action(t_{k}) = \xpar[\big]{\tilde \action(t_{k + 1}) - \tilde \action(t_{k})} / (t_{k + 1} - t_{k})$, for $t_{1} = 0.0$, $t_{2} = 0.2$, \dots, $t_{K + 1} = 1.0$.
The performance of the \gls{SIAR} solutions increases as the set of the side-information is expanded, and near perfect recovery is achieved when, both, the forward invariance of the strategy space and the positive correlation property are imposed as side-information constraints.


\subsection{Using the \glsentryshort{SIAR} Solution to Solve the Equilibrium Selection Problem}
\label{sec:SIAREquilibriumSelectionProblem}
In this \lcnamecref{sec:SIAREquilibriumSelectionProblem}, we use the \gls{SIAR} framework to approximate the game dynamics of an atomic congestion game~$\game$.
We evaluate the performance of the \gls{SIAR} solutions for various update policies, including the ones of replicator dynamics, the log-barrier dynamics, and the smooth Q\nobreakdash-learning dynamics (cf.~\cref{sec:GameDynamics}).
Importantly these game dynamics feature point-wise convergence of their trajectories to stationary points.
In the case of, replicator dynamics and log-barrier dynamics these are Nash equilibria of the game~\parencite{sakos_beating_2024}, while the Q\nobreakdash-learning dynamics converge point-wise to the \gls{QRE} of the game (cf.~\cref{sec:QLearningDynamics}), which are approximately Nash equilibrium points as the exploration rate~$\explorerate$ of the Q\nobreakdash-learning dynamics gets closer to zero.
Furthermore, as the number of players~$\noplayers$ increases, the number of Nash equilibria of $\game$ increase exponentially with respect to $\noplayers$.
Specifically, for the class for two-strategy congestion games we consider the number of equilibrium points is $\abs[\big]{\nash(\game)} = \binom{\noplayers}{\noplayers / 2}$.
This implies that the players are required to solve an exponentially large equilibrium selection problem, and their update policies~$\policy$ correspond to solutions to this problem.
Therefore, good approximations to $\policy$ can capture this solution, and allow a mechanism designer to compute how the players' preferences over the set of Nash equilibria are affected by the system's initial conditions.

As is discussed in~\cref{ex:CongestionGameAnonymity}, atomic congestion games exemplify the property of anonymity of the players, since the rewards function of each player~$i$ depends solely on how many other players have chosen the same strategy as $i$, and not the identities of those players.
For that reason, it is reasonable to assume that the game dynamics satisfy the anonymity of the players, a property that can imposed as side-information constraints in the \gls{SIAR} problem using~\eqref{eq:UpdatePolicyAnonymity}.
Our task in this \lcnamecref{sec:SIAREquilibriumSelectionProblem} is to quantify how close a solution to the \gls{SIAR} problem, where we impose the forward invariance of the state space, the anonymity of the players, and the positive correlation property (given by~\eqref{eq:GeneralizedPositiveCorrelation} in the case of the Q\nobreakdash-learning dynamics) as side-information constraints, are to solve the equilibrium selection problem, i.e., to quantify how close is the point of convergence of the predicted trajectories to the ones of the unknown game dynamics.

We evaluate the \gls{SIAR} framework's average performance on this task using a sample of $100$ \gls{SIAR} solution trained to different random trajectories,
The total number of data~$\nodata = 5$, and the duration of the short-run~$\Delta T = 0.25$ remain fixed for each \gls{SIAR} problem.
The minimum score is $96.1\%$ attained for Q\nobreakdash-learning dynamics and $d = 4$; further details, and a table with the overall scores can be found in \cref{app:AdditonalFiguresAndTables}.


\subsection{Identification of Chaotic Game Dynamics}
\label{sec:ChaoticGameDynamics}
As a final test case, we consider a game of \param{\epsilon}-perturbed \glsxtrlong{RPS} (\param{\epsilon}-\glsxtrshort{RPS}), first studied by \textcite{sato_chaos_2002}. 
The \param{\epsilon}-\glsxtrshort{RPS} is based on a classic \gls{RPS} game, but ties between the two players are broken slightly in favor of one of the players.
Specifically, while any winning combination gives the winning player utility~$1$ and the losing player utility~$- 1$, ties yield utility~$\epsilon$ to the first player, and $- \epsilon$ to the second player, where $- 1 < \epsilon < 1$.
Under the above definition, a classic \gls{RPS} game corresponds to a \param{0}-\glsxtrshort{RPS}.

In a classic rock-paper-scissors game, i.e., \param{0}-\glsxtrshort{RPS}, it is well-known that orbits of the replicator dynamics correspond to cycles around the unique mixed Nash equilibrium of the game.
On the contrary, for any $\epsilon \neq 0$, \textcite{sato_chaos_2002} showed that some orbits exhibit chaos.
The observation was formally proven latter by \textcite{hu_chaotic_2019} by analyzing the spectrum of Lyapunov exponents of the system, which are the rate of separation of any two initially infinitesimally close trajectories.
In the example, we are interested in the \gls{MLE} of the system, given by
\begin{equation}
  \mle 
    \eqdef \sup\set[\bigg]{\lim_{t \to \infty} \lim_{\norm{\delta \cdot u} \to \infty}\frac{1}{t} \ln \frac{\norm{\flow(t, \initialaction) - \flow(t, \initialaction + \delta \cdot u)}}{\norm{\delta \cdot u}} \given \initialaction \in \actions, \,u \in \tangent_{\actions}(\initialaction)}.
\end{equation}
Standard analysis tools can approximate the \gls{MLE} by solving an eigenvalue problem.
The \gls{MLE} is a notion of predictability for the dynamical system; a large positive \gls{MLE} means that a slight estimation error in the system's initialization may lead to radically different states as time evolves.
The timescale defined as the inverse of the \gls{MLE} is called Lyapunov time, and it mirrors the limits in the predictability of the system.

The chaotic nature of the replicator dynamics in a \param{\epsilon}-\glsxtrshort{RPS} game can be visualized through a Poincaré section, i.e., the set of points where an orbit of the system intersects a lower-dimensional subspace.
Following \textcite{sato_chaos_2002}, we consider the subspace $\set{\action \in \actions \given \sum_{i = 1}^{2} \action_{i 1} - \action_{i 2} = 0}$.
In \cref{fig:PoincareSectionsRPS} (left), we plot the Poincaré sections (specifically, their projections into the ($\action_{1 1}, \action_{2, 2}$)\nobreakdash-plane) for different trajectories of the replicator dynamics for a \param{0.25}-\glsxtrshort{RPS} game; we consider the same trajectories as \textcite{sato_chaos_2002}.
\begin{figure}[hb!]
  \centering
  \setkeys{Gin}{width=0.3\textwidth}
  \includegraphics[page=1]{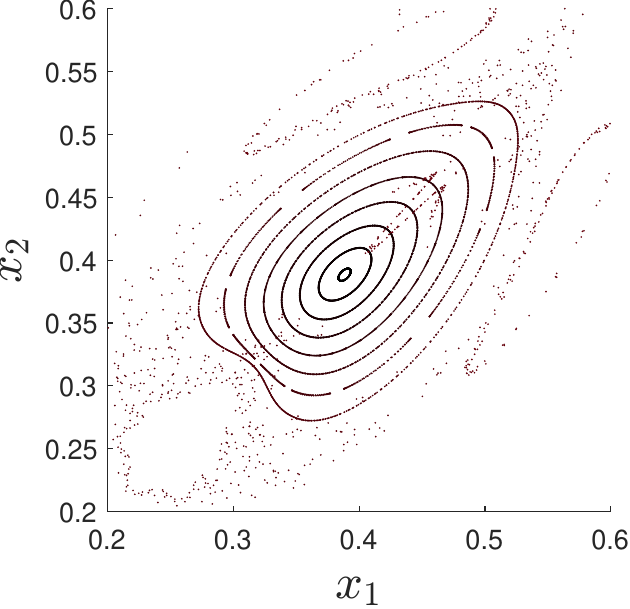}%
  \includegraphics[page=2]{SIAR_perrtubed_rock_paper_scissors}%
  \includegraphics[width=0.35\textwidth, page=4]{SIAR_perrtubed_rock_paper_scissors}
    \caption{%
        On the left, are different Poincaré sections (projected into ($\action_{1 1}, \action_{2, 2}$)\protect\nobreakdash-plane) of the replicator dynamics for a \param{0.25}-\glsxtrshort{RPS} game.
        In the middle, the same Poincaré sections as computed using the \gls{SIAR} solution.
        On the right, is depicted the distance of the corresponding trajectories in Lyapunov times.
    }
    \label{fig:PoincareSectionsRPS}
\end{figure}
In \cref{fig:PoincareSectionsRPS} (middle), we depict the same Poincaré sections as predicted by the \gls{SIAR} solution.
Remarkably, despite the chaotic nature of the replicator dynamics, the sections for the two systems are almost indistinguishable; in \cref{fig:PoincareSectionsRPS} (right), we visualize the distance of those trajectories in Lyapunov times.
The results of the comparison indicate that no significant error ($\geq 0.01$) is observed for at least $8$ Lyapunov times, which is a strong indication of the similarity of the replicator dynamics and the \gls{SIAR} solution.

\section{Discussion \& Future Work}
\label{sec:Discussion}

As systems of interconnected intelligent agents become more pervasive in our daily lives, the ability to identify their dynamical behavior has become crucial for policymakers. 
This allows them to understand the limiting behavior of the system, to control the system to desired outcomes, and to assess counterfactual scenarios of what could have happened under different conditions.
However, current data-driven approaches, which rely on extensive observational data and powerful machine learning algorithms fall short in time-critical scenarios where rapid identification and control of system dynamics are essential to address unexpected events.

In this work, we tackle the challenge of identifying the game dynamics in a strategic multiagent environment using limited observations.
The \glsxtrfull{SIAR} framework introduced in this work exploits the natural game-theoretic traits of such systems and discovers the true dynamics with high accuracy.
Our work brings together tools, techniques, and insights from game theory, learning in games as well as \gls{SOS} optimization, and raises several interesting directions for future work. 
A rather intriguing step is, of course, to move from system identification to actual control of the system.
Specifically, it would be rather beneficial if one could efficiently introduce ``small'' changes to the dynamical behavior such that it results in desirable outcomes, e.g., reduction of the social costs, equilibrium selection, control of chaos, etc.
Bringing such considerations in the current mix of ideas holds the promise of many exciting applications as well as the opportunity for cross-fertilization between traditionally separate fields such as mechanism design, control theory, and multiplayer learning.

\begin{credits}
  \subsubsection{\ackname}
  This research is supported by the MOE Tier 2 grant (MOE-T2EP20223-0018) and by the National Research Foundation, Singapore, and A*STAR under the Quantum Engineering Programme (NRF2021-QEP2-02-P05).
\end{credits}

\newpage
\appendix
\renewcommand{\theequation}{A\arabic{equation}}

\printabbreviations

\section{Game Dynamics}
\label{app:GameDynamics}


In this \lcnamecref{app:GameDynamics}, we give the formal definitions of the \param{q}-replicator dynamics and smooth Q\nobreakdash-learning dynamics that we use throughout this work.


\subsubsection{\param{q}-Replicator Dynamics}
\label{sec:qReplicatorDynamics}
The \param{q}-replicator dynamics \parencite{giannou_rate_2021} are a family of game dynamics that are described by the parametrized update policies $\policy \from \actions \times \nnreals \to \reals^{\nopures_{1}} \times \dots \times \reals^{\nopures_{\noplayers}}$ given by
\begin{equation}
\label{eq:qReplicatorUpdatePolicies}
  \policy_{i j}(\action; q) 
    = \action_{i j}^{q} \xpar[\Bigg]{
        \payoff_{i}(\pure_{i j}, \action_{-i})
      - \frac{
        \sum_{k = 1}^{\nopures_{i}} \action_{i k}^{q} \payoff_i(\pure_{i k}, \action_{-i})
      }{
        \sum_{k = 1}^{\nopures_{i}} \action_{i k}^{q}
      }
    }
    \qquad \action \in \actions, \,i \in [\noplayers], \, j \in [\nopures_{i}]. \tag{QRD}
\end{equation}
Special values of the parameter~$q$ lead to some well-known game dynamics with vastly different properties.
If we set $q = 0$, then for any fully mixed initial strategy profile $\action \in \interior \actions$, the game dynamics (cf. \eqref{eq:GameDynamics}) mimic the (Euclidean) projection dynamics given by the update policies
\begin{equation}
\label{eq:PDUpdatePolicies}
  \policy_{i j}(\action; 0) 
    = \begin{cases}
        \payoff_i(\pure_{i j}, \action_{-i}) 
      - \frac{
        \sum_{k \in \supp(\action_i)} \payoff_i(\pure_{i k}, \action_{-i})
      }{
        \abs{\supp(\action_{i})}
      }
        &j \in \supp(\action_{i}) \\
      0
        &\text{otherwise}
    \end{cases}
    \qquad \action \in \actions, \,i \in [\noplayers], \, j \in [\nopures_{i}], \tag{PD}
\end{equation}
where $\supp(\action_{i})$ denotes the support of the mixed strategy~$\action_{i}$ of player~$i$.
Furthermore, by setting $q = 1$ we recover the replicator dynamics given by the update polices
\begin{equation}
\label{eq:RDUpdatePolicies}
  \policy_{i j}(\action; 1)
    = \action_{i j} \xpar[\big]{
        \payoff_{i}(\pure_{i j}, \action_{-i}) 
      - \payoff_i(\action)
    }
    \qquad \action \in \actions, \,i \in [\noplayers], \,j \in [\nopures_{i}]. \tag{RD}
\end{equation}
Finally, for $q = 2$, one recovers the log-barrier (or inverse update) dynamics given by the update policies 
\begin{equation}
  \policy_{i j}(\action; 2)
    = \action_{i j}^{2} \xpar[\Bigg]{
        \payoff_i(\pure_{i j}, \action_{-i}) 
      - \frac{
        \sum_{k = 1}^{\nopures_{i}} \action_{i k}^{2} \payoff_i(\pure_{i k}, \action_{-i})
      }{
        \sum_{k = 1}^{\nopures_{i}} \action_{i k}^{2}
      }
    }
    \qquad \action \in \actions, \,i \in [\noplayers], \, j \in [\nopures_{i}]. \tag{LBD}
\end{equation}


\subsubsection{Smooth Q-Learning Dynamics} 
\label{sec:QLearningDynamics}
The smooth Q\nobreakdash-learning dynamics for normal-form games, as given by \textcite{tuyls_evolutionary_2006}, are a continuous-time approximation of the Q\nobreakdash-learning algorithm, and are fundamental in economics and artificial intelligence \parencite{camerer_experience-weighted_1999, sutton_reinforcement_2018}.
The Q\nobreakdash-learning dynamics are given by the parametrized update policies
\begin{equation}
\label{eq:QLearningDynamicsUpdatePolicies}
  \policy_{i j}(\action; \explorerate)
    = \action_{i j} \xpar[\Bigg]{
        \payoff_{i}(\pure_{i j}, \action_{-i}) 
      - \payoff_{i}(\action) - \explorerate \xpar[\bigg]{
        \ln \action_{i j} - \sum_{k = 1}^{\nopures_{i}} \action_{i k} \ln \action_{i k}
      }
    }
    \qquad \action \in \actions, \,i \in [\noplayers], \,j \in [\nopures_{i}], \tag{QLD}
\end{equation}
where the parameter~$\explorerate \geq 0$ is called the exploration rate of the players.
Notice that if $\explorerate = 0$, the Q\nobreakdash-learning dynamics are equivalent to the replicator dynamics.
The role of a positive explore rate~$\explorerate$ is to bound the players' rationality, in the sense that, as $\explorerate \to \infty$, the players choose their strategies uniformly at random disregarding the values of their reward functions.

\section{Games with Anonymous Players}
\label{app:PlayerAnonymity}


In this short \lcnamecref{app:PlayerAnonymity}, we show that all atomic congestion games are anonymous games as defined in \eqref{eq:GameAnonymity}.
This is a well-known fact, and the example is meant to help build the proper intuition about this property. 
\begin{example}[An example of normal-form games with anonymous players]
  \label{ex:CongestionGameAnonymity}
    Probably, the most well-known examples of normal-form games with anonymous players are the atomic congestion games.
    In an atomic congestion game~$\game$, the players have identical strategy spaces $\pures_{1} = \dots = \pures_{\noplayers}$ corresponding to possible paths from some source A to some destination B.
    Furthermore, in any pure strategy profile $\pure \in \pures$, each player~$i$ that chooses strategy $\pure_{i} \in \pures_{i}$, exhibits negative reward $\payoff_{i}(\pure_{i}, \pure_{-i})$ that depends only on the total number of players that also choose $\pure_{i}$ (recall that, by definition, $\pures_{1} = \dots = \pures_{\noplayers}$).
    Specifically, for each player~$i$, $\payoff_{i}$ is given by
    \begin{equation}
      \payoff_{i}(\pure_{i}, \pure_{-i})
        = - \sum_{k = 1}^{\noplayers} \one_{\pure_{i} = \pure_{k}}
          \quad \pure \in \pures,
    \end{equation}
    where $\one_{\pure_{i} = \pure_{k}} = 1$ whenever $\pure_{i} = \pure_{k}$, and it's equal to zero otherwise.
  
    Let us verify that~$\game$ satisfies~\eqref{eq:GameAnonymity}.
    Let $\pi \in \perms(\noplayers)$ be some arbitrary permutation function.
    Notice that, since $\pi$ is a permutation function, for each player $k$, by definition, there must exist some player $k'$ such that $\pi(k') = k$.
    Using this observation, we have that, for all players~$i$ and pure strategy profiles~$\pure$, it holds
    \begin{equation}
      \payoff_{\pi(i)}(\pure)
        = - \sum_{k = 1}^{\noplayers} \one_{\pure_{\pi(i)} = \pure_{k}} 
        = - \sum_{k' = 1}^{\noplayers} \one_{\pure_{\pi(i)} = \pure_{\pi(k')}} 
        = - \sum_{k = 1}^{\noplayers} \one_{\pi(\pure)_{i} = \pi(\pure)_{k'}} 
        = \payoff_{i}\xpar[\big]{\pi(\pure)},
    \end{equation}
    where in the second to last equality we used the definition of $\pi(\pure)_{i}$.
    Thus, the atomic congestion game~$\game$ is a game with anonymous players.
  \end{example}

\section{An Example of a \glsentryshort{SIAR} Problem}
\label{app:SIARExample}


In the following example, we show how one can search some unknown update policies~$\policy$ that satisfy the positive correlation side-information constraints given in~\eqref{eq:PositiveCorrelation}.\footnote{
  The conditions of the positive correlation property we use in this example are, first, relaxed as described in \cref{rem:ConstraintsRelaxation}.
}
The goal of this example is to demonstrate how a \gls{SIAR} problem can be constructed in practice, and is meant to aid practitioners in taking advantage of the \gls{SIAR} framework.
In latter \lcnamecrefs{app:SIARExample}, we further show how we can program and solve this problem using the Julia language. 
\begin{example}[An example of a \glsentryshort{SIAR} problem]
\label{ex:SIARProblem}
  Let us consider a game of matching pennies~$\game$ between two players. 
  Their reward functions are given by $\payoff_{1}(\pure) = - \payoff_{2}(\pure) = \payoffmat_{\pure_{1} \pure_{2}}$, for all pure strategy profiles $\pure \in \pures$, where $\payoffmat = \begin{psmallmatrix*}[r] 1 & -1 \\ -1 & 1 \end{psmallmatrix*}$ is the game's payoff matrix.
  Furthermore, let the update policies $\policy \from \actions \to \reals^{2} \times \reals^{2}$ of the players be unknown.
    
  Our goal is to construct a \gls{SIAR} problem to find $\policy$ assuming the game dynamics satisfy the positive correlation property.
  By the above, we have that the vectors of instantaneous payoffs of the players are given, for all strategy profiles~$\action$, by:
  \begin{equation}
    \instantpayoff_{1}(\action) = \begin{pmatrix}
      \action_{2 1} - \action_{2 2} \\
      \action_{2 2} - \action_{2 1}
    \end{pmatrix}
    \quad \text{and} \quad 
    \instantpayoff_{2}(\action) = \begin{pmatrix}
      \action_{1 2} - \action_{1 1} \\
      \action_{1 1} - \action_{1 2}
    \end{pmatrix}.
  \end{equation}
  Substituting the above in~\eqref{eq:PositiveCorrelation} we have that the positive correlation property in $\game$ is given by
  \begin{subequations}
  \begin{alignat}{2}
    (\policy_{1 1}(\action) - \policy_{1 2}(\action)) (\action_{2 1} - \action_{2 2})
      &\geq 0
        &&\qquad \forall \action \in \actions \\
    (\policy_{2 1}(\action) - \policy_{2 2}(\action)) (\action_{1 2} - \action_{1 1}) 
      &\geq 0
        &&\qquad \forall \action \in \actions.
  \end{alignat}
  \end{subequations}
  Next, using the generic form in~\eqref{eq:SIAR} and the above inequalities, we have that the \gls{SIAR} problem in our example is given by
  \begin{equation}
  \begin{aligned} 
    \underset{p_{1}, p_{2}}{\text{minimize}} 
      &&&\sum_{i = 1}^{2} \sum_{k = 1}^{\nodata} \norm{p_{i}\xpar[\big]{\action(t_{k})} - \dot \action_{i}(t_{k})}^{2} \\
    \text{subject to}
      &&&\begin{alignedat}[t]{2}
        &p_{i} \in \polys^{2}[\action]
          &&\qquad i = 1, 2 \\
        &(p_{1 1}(\action) - p_{1 2}(\action)) (\action_{2 1} - \action_{2 2}) \geq 0
          &&\qquad \forall \action \in \actions \\
        &(p_{2 1}(\action) - p_{2 2}(\action)) (\action_{1 2} - \action_{1 1}) \geq 0
          &&\qquad \forall \action \in \actions.
      \end{alignedat}
  \end{aligned} 
  \end{equation}
  Finally, by replacing the nonnegativity constraints with Putinar-type certificates we get the formulation
  \begin{equation}
  \label{eq:SIARProblemRelaxationExample}
  \begin{aligned} 
    \underset{p_{1}, p_{2}, q, \sos}{\text{minimize}} 
      &&&\sum_{i = 1}^{2} \sum_{k = 1}^{\nodata} \norm{p_{i}\xpar[\big]{\action(t_{k})} - \dot \action_{i}(t_{k})}^{2} \\
    \text{subject to}
      &&&\begin{alignedat}[t]{2}
        &p_{i} \in \polys_{d}^{2}[\action], \ q_{i} \in \polys_{d}^{2}[\action], \ \sos_{i} \in \sospolys_{2d}^{5}[\action]
          &&\qquad i = 1, 2 \\
        &(p_{1 1}(\action) - p_{1 2}(\action)) (\action_{2 1} - \action_{2 2}) = g_{1}(\action)
          &&\qquad \forall \action \in \reals^{2} \times \reals^{2} \\
        &(p_{2 1}(\action) - p_{2 2}(\action)) (\action_{1 2} - \action_{1 1}) = g_{2}(\action)
          &&\qquad \forall \action \in \reals^{2} \times \reals^{2},
      \end{alignedat}
  \end{aligned} 
  \end{equation}
  where $g_{i} \from \reals^{2} \times \reals^{2} \to \reals$ is given, for each $i$, by
  \begin{equation}
      g_{i}(\action) 
          = \sos_{i 0}(\action) + \sum_{k = 1}^{2} \sum_{j = 1}^{2} \action_{k j} \cdot \sos_{i k j}(\action) + \sum_{k = 1}^{2} (1 - \action_{k 1} - \action_{k 2}) \cdot q_{i k}(\action),
  \end{equation}
  and $\polys_{d}^{2}[\action]$, $\sospolys_{2d}^5[\action]$ denote, respectively,  the sets of degree~$d$ polynomial vector fields of size~$2$ and degree~$2d$ \gls{SOS} polynomial vector fields of size~$5$ and indeterminates~$\action$.
  For any fixed-degree~$d$, the above formulation is a \gls{SOS} problem, and therefore it can be solved using semidefinite programming.
\end{example}

Having established the above hierarchy of \gls{SOS} problems, let us now see what is their solutions assuming the unknown game dynamics are in fact the replicator dynamics.
In the following example, we construct the exact solution that corresponds to replicator dynamics for the matching pennies game.
\begin{example}[An example of a \glsentryshort{SIAR} solution]
  \label{ex:SIARProblemSolution}
    Continuing our analysis of the game of matching pennies from \cref{ex:SIARProblem}, let us assume that the unknown game dynamics are the replicator dynamics, whose update policies $\policy \from \actions \to \reals^{2} \times \reals^{2}$ are given by~\eqref{eq:RDUpdatePolicies}.
    Substituting in~\eqref{eq:RDUpdatePolicies} the games reward functions, we get that
    \begin{subequations}
    \begin{align}
      \policy_{1 1}(\action)
        &= \action_{1 1} \xpar[\big]{\action_{2 1} - \action_{2 2} - \action_{1 1} (\action_{2 1} - \action_{2 2}) - \action_{1 2} (\action_{2, 2} - \action_{2 1})} \notag \\
        &= \action_{1 1} \xpar[\big]{1 - (\action_{1 1} - \action_{1 2})} (\action_{2 1} - \action_{2 2})
    \intertext{
      for all strategy profiles $\action$.
      Similarly, for all $x$, we have:
    }
      \policy_{1 2}(\action) 
        &= \action_{1 2} \xpar[\big]{1 - (\action_{1 2} - \action_{1 1})} (\action_{2 2} - \action_{2 1}) \\
      \policy_{2 1}(\action) 
        &= \action_{2 1} \xpar[\big]{1 - (\action_{2 1} - \action_{2 2})} (\action_{1 2} - \action_{1 1}) \\
      \policy_{2 2}(\action) 
        &= \action_{2 2} \xpar[\big]{1 - (\action_{2 2} - \action_{2 1})} (\action_{1 2} - \action_{1 2}).
    \end{align}
    \end{subequations}
    Observe that $\policy$ is a polynomial vector field of degree~$3$.
    Now, we can show that, for all points~$\action$ in the ambient space $\reals^{\nopures_{1}} \times \dots \times \reals^{\nopures_{\noplayers}}$, we have
    \begin{subequations}
    \begin{align}
      \inner[\big]{\policy_{1}(\action)}{\instantpayoff_{1}(\action)} 
        &= \action_{1 1} \xpar[\big]{1 - (\action_{1 1} - \action_{1 2})} (\action_{2 1} - \action_{2 2})^{2} + \action_{1 2} \xpar[\big]{1 - (\action_{1 2} - \action_{1 1})} (\action_{2 2} - \action_{2 1})^{2} \notag \\
        &= \action_{1 1} \xpar[\big]{1 - (\action_{1 1} - \action_{1 2})}^{2} (\action_{2 1} - \action_{2 2})^{2} + \action_{1 2} \xpar[\big]{1 - (\action_{1 2} - \action_{1 1})}^{2} (\action_{2 2} - \action_{2 1})^{2} \notag \\
          &\qquad + (1 - \action_{1 1} - \action_{1 2}) (\action_{1 1} - \action_{1 2})^{2} (\action_{2 1} - \action_{2 2})^{2},
    \intertext{and, similarly}
      \inner[\big]{\policy_{2}(\action)}{\instantpayoff_{2}(\action)} 
        &= \action_{2 1} \xpar[\big]{1 - (\action_{2 1} - \action_{2 2})}^{2} (\action_{1 2} - \action_{1 1})^{2} + \action_{2 2} \xpar[\big]{1 - (\action_{2 2} - \action_{2 1})}^{2} (\action_{1 2} - \action_{1 2})^{2} \notag \\
          &\qquad + (1 - \action_{2 1} - \action_{2 2}) (\action_{1 1} - \action_{1 2})^{2} (\action_{2 1} - \action_{2 2})^{2},
    \end{align}
    \end{subequations}
    both of which are $5$\nobreakdash-degree Putinar-type decompositions over $\actions$.
    These equivalences are special cases of a more general result we present in \cref{app:ReplicatorDynamicsProperties}.
      
    By combining the above observations, we get that for $d = 3$ the replicator dynamics is a solution of the \gls{SOS} relaxation of the \gls{SIAR} problem given in~\eqref{eq:SIARProblemRelaxationExample}.
    In particular, one can recover the update policies~$\policy$ of the replicator dynamics, along with the Putinar-type certificates corresponding to positive correlation property, by setting the non-zero unknowns of the problem as follows:
    \begin{subequations}
    \begin{align}
      p_{1 1}(\action) 
        &= \action_{1 1} \xpar[\big]{1 - (\action_{1 1} - \action_{1 2})} (\action_{2 1} - \action_{2 2}) \\
      p_{1 2}(\action) 
        &= \action_{1 2} \xpar[\big]{1 - (\action_{1 2} - \action_{1 1})} (\action_{2 2} - \action_{2 1}) \\
      p_{2 1}(\action) 
        &= \action_{2 1} \xpar[\big]{1 - (\action_{2 1} - \action_{2 2})} (\action_{1 2} - \action_{1 1}) \\
      p_{2 2}(\action) 
        &= \action_{2 2} \xpar[\big]{1 - (\action_{2 2} - \action_{2 1})} (\action_{1 2} - \action_{1 2}) \\
      q_{1 1}(\action)
        &= (\action_{1 1} - \action_{1 2})^{2} (\action_{2 1} - \action_{2 2})^{2} \\
      q_{2 2}(\action)
        &= (\action_{1 1} - \action_{1 2})^{2} (\action_{2 1} - \action_{2 2})^{2} \\
      \sos_{1 1 1}(\action) 
        &= \xpar[\big]{1 - (\action_{1 1} - \action_{1 2})}^{2} (\action_{2 1} - \action_{2 2})^{2} \\
      \sos_{1 1 2}(\action) 
        &= \xpar[\big]{1 - (\action_{1 2} - \action_{1 1})}^{2} (\action_{2 2} - \action_{2 1})^{2} \\
      \sos_{2 2 1}(\action) 
        &= \xpar[\big]{1 - (\action_{2 1} - \action_{2 2})}^{2} (\action_{1 2} - \action_{1 1})^{2} \\
      \sos_{2 2 2}(\action) 
        &= \xpar[\big]{1 - (\action_{2 2} - \action_{2 1})}^{2} (\action_{1 2} - \action_{1 2})^{2}.
    \end{align}
    \end{subequations}
  \end{example}

\section{Proof of \Cref{thm:SideInformationConstraintsMetrics}}
\label{app:ApproximationGuaranteesProofs}


We, now give the proof to \cref{thm:SideInformationConstraintsMetrics} which we restate below for reference.
\SideInformationConstraintsMetrics*
\begin{proof}
  The forward-invariance of the strategy space in~\eqref{eq:StrategySpaceForwardInvariance}, symmetries of the game, e.g., the players' anonymity in~\eqref{eq:UpdatePolicyAnonymity}, and the Nash stationarity in~\eqref{eq:NashStationarity}, are special cases of the side-information constraints analyzed by \textcite[Section 5.2]{ahmadi_learning_2023}, and therefore we refer the interested reader to the corresponding proofs.
  For the rest of the side-information constraints presented in \cref{sec:SideInformationConstraints}, are going to be treated in the remainder of the proof.
    
  Before we continue, to ease the notation, we define $\uncoupleds \eqdef \condiffs\xpar[\big]{\actions_{1} \times \instantpayoff_{1}(\actions)} \times \dots \times \condiffs\xpar[\big]{\actions_{\noplayers} \times \instantpayoff_{\noplayers}(\actions)}$ to be the set of continuously differentiable uncoupled update policies.
  Moreover, in the following formulas, unless otherwise specified, we assume $\action \in \actions$, $i \in [\noplayers]$, and $j \in [\nopures_{i}]$.
  Finally, let us recall that for any functions $h_{1}, h_{2} \condiffs(\actions)$, we have the following inequalities:
  \begin{gather}
    \max_{\action} h_{1}(\action) - \max_{\action} h_{2}(\action)
      \leq \max_{\action} \abs[\big]{h_{1}(\action) - h_{2}(\action)} \label{eq:MaxDiffBound}\\
    \max_{\action} \abs[\big]{h_{1}(\action)} - \max_{\action} \abs[\big]{h_{2}(\action)}
      \leq \max_{\action} \abs[\big]{h_{1}(\action) - h_{2}(\action)} \label{eq:MaxAbsDiffBound},
  \end{gather}
  whose proofs can be found in any standard textbook of real analysis.

  We begin with the existence of an uncoupled formulation that can be imposed by the side-information constraints in~\eqref{eq:UncoupledUpdatePolicy}.
  For this case, we set $L_{S}$ to
  \begin{equation}
    L_{S}(f) 
      = \inf_{f' \in \uncoupleds} \max_{\action, i, j} \abs[\Big]{f_{i j}(\action) - f'_{i j}\xpar[\big]{\action_{i}, \instantpayoff_{i}(\action)}}
        \qquad f \in \condiffs(\actions).
  \end{equation}
  The first condition is, clearly, satisfied by $L_{S}$. 
  In regard to the second condition, let us consider some $\delta > 0$, and two functionals $f, g \in \condiffs(\actions)$.
  Without any loss of the generality, we may also assume that $L_{S}(f) \geq L_{S}(g)$.
  We define the functional $Q \from \condiffs(\actions) \times \uncoupleds \to \nnreals$ given by
  \begin{equation}
    Q(h, h') = \max_{\action, i, j} \abs[\Big]{h_{i j}(\action) - h'_{i j}\xpar[\big]{\action_{i}, \instantpayoff_{i}(\action)}}
      \qquad h \in \condiffs(\actions), \,h' \in \uncoupleds.
  \end{equation}
  Since $\actions$ is a compact set, $\instantpayoff_{1}, \dots, \instantpayoff_{\noplayers}$ are continuous functions (cf.~\eqref{eq:InstantaneousPayoffs}), and $h, h'$ are continuous, it follows that $Q$ is continuous.
  Moreover, since $Q$ is nonnegative, it is bounded from below, and therefore there exists a sequence $h'_{1}, h'_{2}, \dots$ of functionals in $\uncoupleds$ that converge to $\inf_{g' \in \uncoupleds} Q(g, g') = L_{S}(g)$.
  Thus, for~$k$ large enough, we have $Q(g, h'_{k}) - L_{S}(g) \leq \tfrac{\delta}{2}$.
  In the light of the above, we may write
  \begin{subequations}
  \begin{align*}
    \abs[\big]{L_{S}(f) - L_{S}(g)}
      &= L_{S}(f) - L_{S}(g)
        &&\text{as $L_{S}(f) \geq L_{S}(g)$} \\
      &\leq \tfrac{\delta}{2} + L_{S}(f) - Q(g, h'_{k}) \\
      &\leq \tfrac{\delta}{2} + Q(f, h'_{k})- Q(g, h'_{k})
        &&\text{as $L_{S}(f) = \inf_{f' \in \uncoupleds} Q(f, f')$} \\
      &\leq \tfrac{\delta}{2} + \max_{\action, i, j} \abs[\big]{f_{i j}(\action) - g_{i j}(\action)}
        &&\text{by~\eqref{eq:MaxAbsDiffBound}} \\
      &= \tfrac{\delta}{2} + \max_{\action \in \actions} \norm[\big]{f(\action) - g(\action)}_{\infty} \\
      &\leq \tfrac{\delta}{2} + \max_{\action \in \actions} \norm[\big]{f(\action) - g(\action)}
        &&\text{as $\forall y: \norm{y}_{\infty} \leq \norm{y}$}.
  \end{align*}
  \end{subequations}
  Hence, setting $\epsilon = \delta / 2$, and assuming $\max_{\action \in \actions} \norm[\big]{f(\action) - g(\action)} \leq \epsilon$, guarantee that the second condition is satisfied.

  Next, we consider the positive correlation property as given by the side-information constraints in~\eqref{eq:PositiveCorrelation}, and set $L_{S}$ to
  \begin{equation}
      L_{S}(f) 
        = \max_{\action \in \actions} \set[\big]{0, - \inner[\big]{f(\action)}{\instantpayoff(\action)}}
          \qquad f \in \condiffs(\actions),
  \end{equation}
  where for each action profile~$\action$, we write $\instantpayoff(\action)$ to denote the ensemble of the instantaneous payoffs $\instantpayoff_{1}(\action)$, \dots, $\instantpayoff_{\noplayers(\action)}$.
  Once again, first condition is clearly satisfied by $L_{S}$.
  Moreover, for any $\delta > 0$, and any functionals $f, g \in \condiffs(\actions)$ such that, without any loss of the generality, $L_{S}(f) \ge L_S(g)$, we have two cases: either
  \begin{enumerate*}
    \item $L_{S}(\policy) = 0$; or
    \item $L_S(\policy) > 0$.
  \end{enumerate*}
  In the former case, since $L_{S}(f) \ge L_{S}(g) \geq 0$, we have $L_{S}(f) = L_{S}(g) = 0$, and therefore the second condition holds for any $\epsilon > 0$.
  In the latter case, we have that 
  \begin{subequations}
  \begin{align*}
    \abs[\big]{L_{S}(f) - L_{S}(g)}
        &= L_{S}(f) - L_{S}(g)
          &&\text{as $L_{S}(f) \geq L_{S}(g)$} \\
        &= \max_{\action} \xpar[\big]{- \inner[\big]{f(\action)}{\instantpayoff(\action)}} - \max_{\action} \set{0, - \inner[\big]{g(\action)}{\instantpayoff(\action)}} 
          &&\text{as $L_{S}(f) > 0$} \\
        &\leq \max_{\action \in \actions} \xpar[\big]{- \inner[\big]{f(\action)}{\instantpayoff(\action)}} - \max_{\action} \xpar[\big]{- \inner{g(\action)}{\instantpayoff(\action)}} 
          &&\text{as $L_{S}(g) \geq \max_{\action} \xpar[\big]{- \inner{g(\action)}{\instantpayoff(\action)}}$} \\
        &\leq \max_{\action} \abs[\Big]{\inner[\big]{f(\action)}{\instantpayoff(\action)} - \inner[\big]{g(\action)}{\instantpayoff(\action)}} 
          &&\text{by~\eqref{eq:MaxDiffBound}} \\
        &= \max_{\action} \abs[\Big]{\inner{f(\action) - g(\action)}{\instantpayoff(\action)}} \\
        &\leq \max_{\action} \norm{f(\action) - g(\action)} \cdot \norm{\instantpayoff(\action)} 
          &&\text{by \glsxtrshort{CS}} \\ 
        &= \max_{\action} \norm{f(\action) - g(\action)} \cdot \max_{\action} \norm{\instantpayoff(\action)},
  \end{align*}
  \end{subequations}
  Therefore, if we set $\epsilon = \delta / \max_{\action} \norm{\instantpayoff(\action)}$, and assuming $\max_{\action \in \actions} \norm[\big]{f(\action) - g(\action)} \leq \epsilon$, the second condition is satisfied.

  Finally, we remark that, except for minor technical deviations, the proof of the remaining cases of the generalized positive correlation property imposed by~\eqref{eq:GeneralizedPositiveCorrelation} and the property of elimination of strictly dominated strategies imposed by \cref{eq:RegularSelectionDynamicsUpdatePolicy,eq:ConvexMonotonicity} follow similar reasoning as the above, and therefore they are omitted. 
  For reference, we state only the forms of the $L_{S}$ functionals in each of these cases.
  For the former case, given explore rate~$explorerate \geq 0$ constant, we can set $L_{S}$ to
  \begin{equation}
    L_{S}(f) 
      = \max_{\action \in \actions} \set[\big]{0, - \inner[\big]{f(\action)}{\instantpayoff'(\action; \explorerate)}}
        \qquad f \in \condiffs(\actions),
  \end{equation}
  where for each action profile~$\action$, we write $\instantpayoff'(\action; \explorerate)$ to denote the ensemble of the parametrized vector fields $\instantpayoff'_{1}(\action; \explorerate)$, \dots, $\instantpayoff'_{\noplayers(\action; \explorerate)}$ defined as in~\eqref{eq:ParametrizedInstantenuousPayoffs}.
  While, for the latter case, the functional~$L_{S}$ gets the complicated form
  \begin{equation}
    L_{S}(f) 
      = \inf_{g \in \condiffs(\actions)} \max \set*{\begin{alignedat}{1}
        &\max_{\action, i, j} \abs[\big]{f_{i j}(\action) - \action_{i j} \cdot g_{i j}(\action)} \\
        &\max_{\action, i} \abs[\Big]{\inner[\big]{\action_{i}}{g_{i}(\action)}} \\
        &\max_{\mathclap{\substack{\action'_{i} \in \actions_{i} \\ \action, i, j}}} \set[\Big]{
          \set[\big]{0, \inner{\action'_{i}, g_{i}(\action)} - g_{i j}(\action)}
            \given \payoff_{i}(\action'_{i}, \action_{-i}) \geq \payoff_{i}(\pure_{i j}, \action_{-i})
        }
      \end{alignedat}
    }
     \qquad \policy \in \condiffs(\actions).
  \end{equation}
\end{proof}

\section{Properties of Replicator Dynamics as Putinar-Type Decompositions}
\label{app:ReplicatorDynamicsProperties}


In this \lcnamecref{app:ReplicatorDynamicsProperties}, we show that the replicator dynamics, as given in~\eqref{eq:RDUpdatePolicies}, satisfy the side-information constraints presented in~\cref{sec:SideInformationConstraints}.
Furthermore, we show that the nonnegativity conditions of each constraint have a Putinar-type \gls{SOS} decomposition in the set of mixed strategy profiles~$\actions$.

In the sequel, we assume $\game$ to be some arbitrary $\noplayers$\nobreakdash-player normal-form game whose players' behavior is evolving according to the replicator dynamics and is captured by the updated policies $\policy \from \actions \to \reals^{\nopures_{1}} \times \dots \times \reals^{\nopures_{\noplayers}}$.
The fact that $\policy$ satisfies the following properties is well-known in the literature \parencite{sandholm_population_2010}.
However, to the best of our knowledge, no attempt has been made to construct the appropriate Putinar-type decompositions.
The existence of those decompositions certifies that a solution to the \gls{SIAR} problem can recover the replicator dynamics.
In addition to the following properties, replicator dynamics also satisfy the Nash stationarity and are, in fact, uncoupled game dynamics.
Since those two properties are not associated with any nonnegativity constraints (cf. \cref{sec:SideInformationConstraints}), they are excluded by the following analysis.

Now, we state and prove the following \lcnamecrefs{thm:RDForwardInvarianceSOS} about replicator dynamics:

\begin{proposition}[Forward invariance of the strategy space]
\label{thm:RDForwardInvarianceSOS}
  The faces of the strategy space~$\actions$---and thus, the strategy space itself---are forward-invariant with respect to the replicator dynamics.
  Specifically, for all players~$i$ and $\action \in \reals^{\nopures_{1}} \times \dots \times \reals^{\nopures^{\noplayers}}$, it holds
  \begin{subequations}
  \begin{alignat}{2}
    \sum_{j = 1}^{\nopures_{i}} \policy_{i j}(\action)
      &= \xpar[\bigg]{1 - \sum_{j = 1}^{\nopures_{i}}} \cdot p_{i}(\action) \\ 
    \policy_{i j}(\action)
      &= \action_{i j} \cdot q_{i j}(\action)
        &&\quad \forall j \in [\nopures_{i}],
  \end{alignat}
  \end{subequations}
  where $p_{i}, q_{i j}$ are some $\noplayers$\nobreakdash-degree polynomials
\end{proposition}
\begin{proof}
    For the first equality, notice that, for any player~$i$ and $\action \in \reals^{\nopures_{1}} \times \dots \times \reals^{\nopures^{\noplayers}}$, we have
    \begin{equation}
      \sum_{j = 1}^{\nopures_{i}} \policy_{i j}(\action)
        = \sum_{j = 1}^{\nopures_{i}} \action_{i j} \xpar[\big]{\payoff_{i}(\pure_{i j}, \action_{-i}) - \payoff_{i}(\action)}
        = \payoff_{i}(\action) - \payoff_{i}(\action) \sum_{j = 1}^{\nopures_{i}} \action_{i j}
        = \xpar[\bigg]{1 - \sum_{j = 1}^{\nopures_{i}} \action_{i j}} \cdot \payoff_{i}(\action).
    \end{equation}
    Furthermore, the second equality follows trivially by the definition of $\policy_{i, j}$ by setting $q_{i j}(\action) = \payoff_{i}(\pure_{i j}, \action_{-i}) - \payoff_i(\action)$.
    Therefore, setting $p_{i}(\action) = \payoff_{i}(\action)$ and $q_{i j}(\action)$ as above completes the proof.
\end{proof}

\begin{proposition}[Anonymity of the players]
  Given a normal-form game with anonymous players, the replicator dynamics satisfy the player anonymity property.
  Specifically, for any player~$i$, permutation $\pi \from [\noplayers] \to [\noplayers]$, and $\action \in \reals^{\nopures_{1}} \times \dots \times \reals^{\nopures_{\noplayers}}$, we have that
  \begin{equation}
    \policy_{\pi(i)}(\action)
      = \policy_{i}\xpar[\big]{\pi(\action)}
  \end{equation}
  where the mixed strategy profile~$\pi(\action) \in \actions$ is given by $\pi(\action)_{i} \eqdef \action_{\pi(i)}$ for all $i \in [\noplayers]$.
\end{proposition}
\begin{proof}
    As a first step, let us observe that for any permutation $\pi$, player~$i$, mixed action~$\action'_{i} \in \actions_{i}$, and mixed action profile~$\action_{-i} \in \actions_{-i}$, the point $y = \xpar[\big]{\action'_{\pi(i)}, \actions_{-\pi(i)}}$ corresponds to the mixed action profile where player $\pi(i)$ chooses $\action'_{\pi(i)}$ and the rest of the players---possibly, including $i$---choose $\action_{-i}$.
    Therefore, $\pi(y)$ corresponds to the action profile where player~$i$ chooses the action of player~$\pi(i)$ in $y$, i.e., $\pi(y)_{i} = y_{\pi(i)}$.
    Using the above observation, for all permutations~$\pi$, players~$i$, $j \in [\nopures_{i}]$, and $\action \in \reals^{\nopures_{1} \times \dots \times \reals^{\nopures_{\noplayers}}}$, we have that
    \begin{subequations}
    \begin{align}
      \policy_{\pi(i) j}(\action)
        &= \action_{\pi(i) j} \xpar[\big]{\payoff_{\pi(i)}(\pure_{\pi(i) j}, \action_{-\pi(i)}) - \payoff_{\pi(i)}(\action)} \\
        &= \pi(\action)_{i j} \xpar[\Big]{\payoff_{i}\xpar[\big]{\pi(\pure)_{i j}, \pi(\action)_{-i}} - \payoff_{i}\xpar[\big]{\pi(\action)}} \\
        &= \policy_{i j}\xpar[\big]{\pi(\action)},
    \end{align}
  \end{subequations}
    where in the second equality we used the fact that the players are anonymous.
\end{proof}

\begin{proposition}[Positive correlation]
  The replicator dynamics satisfy the positive correlation property.
  Specifically, for all players~$i$ and $\action \in \reals^{\nopures_{1}} \times \dots \times \reals^{\nopures^{\noplayers}}$, it holds
  \begin{equation}
    \inner{\policy_{i}(\action)}{\instantpayoff_{i}\xpar[\big]{\action}} 
      = \sum_{j = 1}^{\nopures_{i}} \action_{i j} \cdot \sos_{i j}(\action) + \xpar[\bigg]{1 - \sum_{j = 1}^{\nopures_{i}} \action_{i j}} \cdot \sos_{i}(\action),
  \end{equation}
  where $\sos_{i j}$ and $\sos_{i}$ are some $2\noplayers$\nobreakdash-degree \gls{SOS} polynomials.
\end{proposition}
\begin{proof}
  Expanding the left-hand-side of the equality we have, for any $i \in [\noplayers]$ and $\action \in \reals^{\nopures_{1} \times \dots \times \reals^{\nopures_{\noplayers}}}$, that
  \begin{subequations}
  \begin{align}
    \inner{\policy_{i}(\action)}{\instantpayoff_{i}\xpar[\big]{\action}}
      &= \sum_{j = 1}^{\nopures_{i}} \action_{i j} \cdot \xpar[\big]{\payoff_{i}(\pure_{i j}, \action_{-i}) - \payoff_{i}(\action)} \cdot \payoff_{i}(\pure_{i j}, \action_{-i}) \\
      &= \sum_{j = 1}^{\nopures_{i}} \action_{i j} \cdot \payoff_{i}(\pure_{i j}, \action_{-i})^{2} - \payoff_{i}(\action)^{2} \\
      &= \sum_{j = 1}^{\nopures_{i}} \action_{i j} \cdot \payoff_{i}(\pure_{i j}, \action_{-i})^{2} - 2 \cdot \payoff_{i}(\action)^{2} + \payoff_{i}(\action)^{2} \\
      &= \sum_{j = 1}^{\nopures_i} \action_{i j} \cdot \payoff_i(\pure_{i j}, \action_{-i})^{2} - 2 \cdot \payoff_{i}(\action)^{2} + \payoff_{i}(\action)^{2} \sum_{j = 1}^{\nopures_{i}} \action_{i j} + \xpar[\bigg]{1 - \sum_{j = 1}^{\nopures_{i}} \action_{i j}} \cdot \payoff_{i}(\action)^{2} \\
      &= \sum_{j = 1}^{\nopures_i} \action_{i j} \xpar[\big]{\payoff_{i}(\pure_{i j}, \action_{-i})^{2} - 2 \cdot \payoff_{i}(\pure_{i j}, \action_{-i}) \cdot \payoff_{i}(\action) + \payoff_{i}(\action)^{2}} + \xpar[\bigg]{1 - \sum_{j = 1}^{\nopures_{i}} \action_{i j}} \cdot \payoff_{i}(\action)^{2} \\
      &= \sum_{j = 1}^{\nopures_{i}} \action_{i j} \cdot \xpar[\big]{\payoff_{i}(\pure_{i j}, \action_{-i}) - \payoff_i(\action)}^2 + \xpar[\bigg]{1 - \sum_{j = 1}^{\nopures_{i}} \action_{i j}} \cdot \payoff_{i}(\action)^{2}.
  \end{align}
\end{subequations}
  Thus, the claim holds for $\sos_{i j}(\action) = \xpar[\big]{\payoff_{i}(\pure_{i j}, \action_{-i}) - \payoff_{i}(\action)}^{2}$ and $\sos_{i}(\action) = \payoff_{i}(\action)^{2}$.
  Furthermore, for any $\action \in \actions$ such that $\inner{\policy_{i}(\action)}{\instantpayoff_{i}\xpar[\big]{\action}} = 0$, the above implies that for all players~$i$ and $j \in [\nopures_{i}]$, either $\action_{i j} = 0$ or $\payoff_{i}(\pure_{i j}, \action_{-i}) = \payoff_{i}(\action)$.
  Therefore, $\policy(\action) = 0$.
  It follows that $\inner{\policy_{i}(\action)}{\instantpayoff_{i}\xpar[\big]{\action}}$ is strictly positive for any non-stationary state $\action \in \actions$, which completes the proof. 
\end{proof}

\begin{proposition}[Elimination of strictly dominated strategies]
  The replicator dynamics asymptotically lead to the elimination of strictly dominated actions.
  Specifically, for each player~$i$ and $j \in [\nopures_{i}]$ there exist $\noplayers$\nobreakdash-degree polynomials $g_{i j} \from \reals^{\nopures_{1}} \times \dots \times \reals^{\nopures_{\noplayers}} \to \reals$ such that
  \begin{subequations}
  \begin{alignat}{2}
    \sum_{j = 1}^{\nopures_{i}} \action_{i j} \cdot g_{i j}(\action) 
      &= \xpar[\bigg]{1 - \sum_{j = 1}^{\nopures_{i}} \action_{i j}} \cdot p_{i}(\action) \\
    \policy_{i j}(\action) 
      &= \action_{i j} \cdot g_{i j}(\action)
        &&\quad \forall j \in [\nopures_{i}] \\
    \inner{\action'_{i}}{g_{i}(\action)} - g_{i j}(\action) 
      &= \payoff_{i}(\action'_{i}, \action_{-i}) - \payoff_{i}(\pure_{i j}, \action_{-i}) + \xpar[\bigg]{1 - \sum_{k = 1}^{\nopures_{i}} \action_{i k}} \cdot q_{i, j}(\action)
        &&\quad \forall j \in [\nopures_{i}],
  \end{alignat} 
  \end{subequations}
  for all players~$i$, $\action'_{i} \in \reals^{\nopures_i}$, and $\action \in \reals^{\nopures_{1}} \times \dots \times \reals^{\nopures_{\noplayers}}$, where $p_{i}$ and $q_{i j}$ are some $\noplayers$\nobreakdash-degree polynomials.
\end{proposition}
\begin{proof}
  By setting $g_{i j}(\action) = \payoff_{i}(\pure_{i j}, \action_{-i}) - \payoff_{i}(\action)$ and $p_{i}(\action) = \payoff_{i}(\action)$, the first two equalities hold as a consequence of~\cref{thm:RDForwardInvarianceSOS}.
  Substituting the above in the last equality we also have, for any player~$i$, $j \in [\nopures_{i}]$, $\action'_{i} \in \reals^{\nopures_{i}}$ and $\action \in \reals^{\nopures_{1}} \times \dots \times \reals^{\nopures_{\noplayers}}$, that
  \begin{subequations}
  \begin{align}
    \inner{\action'_{i}}{g_{i}(\action)} - g_{i j}(\action)
      &= \sum_{k = 1}^{\nopures_{i}} \action'_{i k} \xpar[\big]{\payoff_{i}(\pure_{i k}, \action_{-i}) - \payoff_{i}(\action)} - \xpar[\big]{\payoff_{i}(\pure_{i j}, \action_{-i}) - \payoff_{i}(\action)} \\
      &= \payoff_{i}(\action'_{i}, \action_{-i}) - \payoff_{i}(\pure_{i j}, \action_{-i}) + \xpar[\bigg]{1 - \sum_{k = 1}^{\nopures_{i}} \action_{i k}} \cdot \payoff_{i}(\action).
  \end{align}
  \end{subequations}
  Therefore, the last equality holds by additionally setting $q_{i j}(\action) = \payoff_{i}(\action)$.
  Furthermore, for any $\pure_{i j} \in \pures_{i}$ that is strictly dominated by some mixed action profile $\action \in \actions$, the last expression is strictly positive, which implies that $\inner{\action'_{i}, g_i(\action)} - g_{i j}(\action) > 0$ and, thus, the proof is complete.
\end{proof}

\section{Additional Figures \& Tables}
\label{app:AdditonalFiguresAndTables}


We evaluate the \gls{SIAR} framework's average performance on this task using a sample of $100$ \gls{SIAR} solution trained to different random trajectories,
The total number of data~$\nodata = 5$, and the duration of the short-run~$\Delta T = 0.25$ remain fixed for each \gls{SIAR} problem.
Subsequently, each solution is evaluated on $1000$ random trajectories which are uniformly initialized in strategy space.
An evaluation is considered successful if the strategy profile at time~$t = 200$ lies $\epsilon = 0.01$ close to the Nash equilibrium (\gls{QRE}, in the case of the Q\nobreakdash-learning dynamics) that the unknown game dynamics converge to for the same initialization.
The parameters $t$ and $\epsilon$ remain constant across all evaluations.
Additionally, in the case that the unknown game dynamics is the smooth Q\nobreakdash-learning dynamics, we do not assume knowledge of their exploration rate~$\explorerate = 0.2$.
Instead, we search for an exploration rate~$\explorerate \in [0, 1]$ that best fits the observed data using binary search.
Overall, the \gls{SIAR} solutions exhibit remarkable accuracy as the degree~$d$ of the polynomial solutions becomes larger even in the case the unknown game dynamics are non-polynomial, e.g., the log-barrier dynamics and Q\nobreakdash-learning dynamics.
The performance of the \gls{SIAR} solutions in two-strategy congestion games with different number of players~$\noplayers \in \set{2, 4, \dots, 10}$ is summarized in \cref{tbl:SIARPerformanceEquilbriumSelectionProlem}.
\begin{table}[hb!]
  \centering
      \caption{%
          Average accuracy (\%) of the \gls{SIAR} solution of degree~$d$ in the equilibrium selection problem for $\noplayers$\protect\nobreakdash-player two-strategy congestion games with $\binom{\noplayers}{\noplayers / 2}$ of Nash equilibria; the accuracy is rounded to one decimal place.
          The \gls{SIAR} solutions capture the limiting behavior of the game dynamics even if the update policies are non-polynomial.
      }
      \sisetup{table-format=3.1}
      \begin{tabular}{@{}S[table-format=2]SSSSSSS@{}} 
      \toprule
            \multicolumn{1}{c|}{}
          & \multicolumn{3}{|c|}{RD}
          & \multicolumn{2}{|c|}{LBD}
          & \multicolumn{2}{|c}{QLD} \\ 
            \multicolumn{1}{c|}{$\noplayers$}
          & \multicolumn{1}{|c}{$d = 3$}
          & \multicolumn{1}{c}{$d = 4$}
          & \multicolumn{1}{c|}{$d = 5$}
          & \multicolumn{1}{|c}{$d = 4$}
          & \multicolumn{1}{c|}{$d = 5$}
          & \multicolumn{1}{|c}{$d = 4$}
          & \multicolumn{1}{c}{$d = 5$} \\
      \midrule
           2 & 100 & 100 &  99.99 &  99.99 &  99.99 &  96.12 &  99.99 \\   
           4 & 100 & 100 & 100    & 100    & 100    &  99.82 & 100    \\   
           6 & 100 & 100 & 100    & 100    & 100    & 100    & 100    \\   
           8 & 100 & 100 & 100    & 100    & 100    & 100    & 100    \\   
          10 & 100 & 100 & 100    &  99.98 &  99.97 & 100    & 100    \\  
      \bottomrule
      \end{tabular}
      \label{tbl:SIARPerformanceEquilbriumSelectionProlem}
  \end{table} 

\printbibliography

\end{document}